\documentclass[%
reprint,
superscriptaddress,
frontmatterverbose, 
preprintnumbers,
nofootinbib,
longbibliography,
 amsmath,amssymb,
 aps,
 pra,
notitlepage
]{revtex4-1}

\usepackage{amsmath,amsthm,amsfonts,graphicx}
\usepackage{amssymb}
\usepackage{epsfig,graphicx,color}
\usepackage{mathrsfs}
\usepackage[active]{srcltx}
\usepackage{url}
\usepackage{enumerate}
\usepackage{hyperref}

\usepackage{ytableau}
\usepackage{caption}
\usepackage{subcaption}

\hypersetup{
colorlinks=true,
linkcolor=blue,
filecolor=blue,
citecolor=blue,  
urlcolor=blue,
}
\usepackage{comment}

\newcommand{\R}{\mathbb{R}}
\newcommand{\C}{\mathbb{C}}

\renewcommand{\H}{\mathcal{H}}

\def\>{\rangle}
\def\<{\langle}

\newcommand{\Tr}{\operatorname{Tr}}

\newtheorem{theo}{Theorem}

\newtheorem{lemma}{Lemma}
\newtheorem{prop}[theo]{Proposition}

\newtheorem{defi}{Definition}

\newcommand{\E}{\mathbb{E}}

\newcommand{\worst}{{\operatorname{worst}}}
\newcommand{\avg}{{\operatorname{avg}}}
\newcommand{\choi}{{\operatorname{Choi}}}
\newcommand{\ovl}{\overline}
\newcommand{\ul}{\underline}
\newcommand{\id}{\operatorname{id}}
\newcommand{\sgn}{\operatorname{sgn}}
\newcommand{\poly}{\operatorname{poly}}

\newcommand{\cE}{\mathcal{E}}
\newcommand{\cN}{\mathcal{N}}
\newcommand{\cD}{\mathcal{D}}
\newcommand{\cT}{\mathcal{T}}
\newcommand{\cL}{\mathcal{L}}
\newcommand{\cR}{\mathcal{R}}
\newcommand{\cC}{\mathcal{C}}
\newcommand{\cU}{\mathcal{U}}

\newcommand{\cJ}{\mathcal{J}}

\newcommand{\zw}[1]{{\color{black}#1}}

\newcommand{\linghang}[1]{{\color{black}#1}}
\newcommand{\lh}[1]{{\color{black}#1}}
\newcommand{\lhn}[1]{{\color{black}#1}}

\newcommand{\zwmerge}[1]{{\color{black}#1}}
\newcommand{\zwnew}[1]{{\color{black}#1}}
\newcommand{\zwnn}[1]{{\color{black}#1}}
\newcommand{\zwnnn}[1]{{\color{black}#1}}

\begin{document}
\title{Near-optimal covariant quantum error-correcting codes \\from random unitaries with symmetries}




\author{Linghang Kong} \email{linghang@mit.edu}
\affiliation{Center for Theoretical Physics, MIT, Cambridge, MA 02139, United States}
\author{Zi-Wen Liu} \email{zliu1@perimeterinstitute.ca}
\affiliation{Perimeter Institute for Theoretical Physics, Waterloo, Ontario N2L 2Y5, Canada}


\begin{abstract}
Quantum error correction and symmetries play central roles in quantum information science and physics.  It is known that quantum error-correcting codes that obey (are covariant with respect to) continuous symmetries in a certain sense cannot correct erasure errors perfectly (a well-known result in this regard being the Eastin--Knill theorem \zwnn{in the context of fault-tolerant quantum computing}), in contrast to the case without symmetry constraints. Furthermore, several quantitative fundamental limits on the accuracy of such covariant codes for approximate quantum error correction are known. Here, we consider the quantum error correction capability of uniformly random covariant codes.  In particular, we analytically study the most essential cases of $U(1)$ and $SU(d)$ symmetries, and show that for both symmetry groups the error of the covariant codes generated by Haar-random symmetric unitaries, i.e.,~unitaries that commute with the \zwnn{group actions}, typically \zwnn{scale as $O(n^{-1})$ in terms of both the average- and worst-case purified distances against erasure noise, saturating}  the fundamental limits to leading order.  We note that the results hold for symmetric variants of unitary 2-designs, and comment on the convergence problem of symmetric random circuits. Our results not only indicate (potentially efficient) randomized constructions of optimal $U(1)$- and $SU(d)$-covariant codes, but also reveal fundamental properties of random symmetric unitaries, which 
yield important solvable models of complex quantum systems (including black holes and many-body spin systems) that have attracted great recent interest in quantum gravity and condensed matter physics.
We expect our construction and analysis to \zwnn{find broad relevance in} both physics and \zwnn{quantum computing}.  
\end{abstract}

\preprint{MIT-CTP/5367}
\maketitle

\section{Introduction}

\zw{
One of the most important and widely studied ideas in quantum information processing is quantum error correction (QEC) \cite{PhysRevA.52.R2493,nielsen2011,gottesman2010,lidar2013}, which protects (logical) quantum systems against noise and errors by suitably encoding them into quantum error-correcting codes living in a larger physical Hilbert space. 
Besides the clear importance to the practical realization of quantum computing and other quantum technologies, QEC and quantum codes have also drawn great interest in physics recently as they are found to arise in many important physical scenarios in e.g.~holographic quantum gravity \cite{AlmheiriDongHarlow,pastawski2015} and many-body physics \cite{KITAEV20032,2015arXiv150802595Z,brandao2019}.

Physical systems typically entail symmetries and conservation laws, which restrict their behaviors in certain fundamental ways.
Therefore,  it is clearly important to understand how symmetry constraints may influence QEC, considering its practical significance and broad physical relevance.  
More explicitly, in the presence of symmetries, the encoders are restricted to be covariant with respect to the symmetry group (i.e.,~commute with certain forms of group actions), generating the so-called \emph{covariant codes} \cite{hayden2017,faist2019,Woods2020continuousgroupsof}.  
Covariant codes are known to have broad relevance in both practical and theoretical aspects, arising in many important areas in quantum information and physics such as fault tolerance \cite{eastin2009}, quantum reference frames \cite{hayden2017}, anti-de Sitter/conformal field theory (AdS/CFT) correspondence \cite{HarlowOoguri2018arXiv181005338H,harlow2019,kohler2019,faist2019,Woods2020continuousgroupsof}, and condensed matter physics \cite{brandao2019}.

When the symmetry is continuous (mathematically modelled by a Lie group), there exist fundamental limitations on the QEC capability of the corresponding covariant codes.  A well known no-go theorem in this regard is the Eastin--Knill theorem \cite{eastin2009}, which indicates that codes covariant with respect to continuous symmetries in the sense that the logical group actions are mapped to ``transversal'' physical actions that are tensor products on physical subsystems (which is highly desirable for fault tolerance since they do not spread errors within code blocks) cannot correct {local} errors perfectly (for physical systems with finite Hilbert space dimension).  
A physical interpretation of this phenomenon is that some logical charge information is necessarily leaked into the environment due to the error for such covariant codes, which forbids perfect recovery. 
Then the question naturally arises: to what degree can the QEC task be done approximately under these constraints?  Several ``robust'' versions of the Eastin-Knill theorem giving quantitative lower bounds on the inaccuracy of covariant codes have been recently established \cite{faist2019,Woods2020continuousgroupsof,KubicaD,zlj20,YangWoods,PRXQuantum.3.010337}, some of which employing methods from other areas of independent interest such as quantum clocks \cite{Woods2020continuousgroupsof}, quantum metrology \cite{KubicaD,zlj20,YangWoods}, and quantum resource theory \cite{zlj20,PRXQuantum.3.010337}.  
  
}

\zw{
This work concerns the achievability of such lower bounds.  Here we specifically consider two most fundamental and representative  continuous symmetry groups in quantum mechanics, \zwmerge{$U(1)$ and $SU(d)$.    $U(1)$ is the most basic continuous symmetry associated with the conservation of a single quantity (which may physically correspond to charge, energy, particle number etc.), and $SU(d)$ \zwnnn{represents a key type of non-Abelian symmetry groups associated with non-commuting charges. In particular, since $SU(d)$} describes the entire group of unitary actions on a $d$-dimensional quantum system, 
it is  closely related to  quantum computing.}   \zwmerge{In this work, we consider a natural approach to constructing random covariant codes using unitaries drawn from the Haar measure that commute with the symmetry actions, which is particularly interesting because of the following: i) The results faithfully indicate typical properties of all symmetric unitaries due to the uniform nature of the Haar measure;} ii) Haar-random unitaries and their relatives including unitary designs and random circuits \zwnnn{have become standard tools or models} in the study of complex many-body quantum systems such as black holes \cite{hayden2007,hosur2016,PRXQuantum.2.020339} and chaotic spin systems \cite{Nahum2,Nahum1,PhysRevX.8.021013} \zwnnn{due to their ``scrambling'' but solvable features}, indicating that our refined models of random unitaries with symmetries  are potentially of broad interest in physics \zwnnn{(see also Refs.~\cite{Yoshida:softmode,nakata2021black,liu2020,PhysRevX.8.031057,PhysRevX.8.031058} for some recent studies of relevant models in physical contexts)}.
We rigorously analyze the performance of our random covariant codes against erasure noise, as characterized by both the average-case and the worst-case  recovery error (measured by the purified distance) for all input states.   
To do so, we use the complementary channel technique \cite{beny2010}, which allows one to characterize the error rate of a code by the amount of information leaked into the environment.  At a high level, our derivation of the code errors for both $U(1)$ and $SU(d)$ is based on breaking down the error into two components, one characterizing the deviation of the physical state from its average \lhn{(over the randomness in the encoding)}, which leads to \zwnnn{an} error that can be bounded using a ``partial decoupling'' theorem \cite{wakakuwa2019} and turns out to be exponentially small, while the other characterizing a polynomially small intrinsic error induced by \zwnnn{symmetry}. 
We show that \zwnnn{ our random codes almost always saturate the lower bounds to leading order (up to constant factors, in certain cases exactly)}, indicating that \zwnnn{the symmetric} unitaries typically give rise to nearly optimal covariant codes.   The results hold if the Haar-random \zwmerge{symmetric} unitary is simplified to corresponding 2-designs, and as we conjecture, efficient random circuits composed of \zwmerge{symmetric} local gates.   Note that in our case with \zwmerge{symmetries} the error is intrinsically polynomially small, while in the no-symmetry case the error of such Haar-random codes is normally exponentially small and there exist perfect codes.

}

\zw{
The rest of this paper is structured as follows. In Sec.~\ref{sec:preliminary} we formally introduce the relevant background. 
Then in Sec.~\ref{sec:u1} and Sec.~\ref{sec:sud} we respectively discuss the $U(1)$ and $SU(d)$ cases. In each case, we first define our construction of the random covariant codes based on Haar-random symmetric unitaries, then present the analysis of their error  as measured by Choi and worst-case purified distances, and finally make explicit comparisons  with known lower bounds on the code error and other known constructions.
In Sec.~\ref{sec:designs}, we discuss the extension of our study to symmetric $t$-designs and random circuits.   We conclude the work with discussions on future directions, in particular the potential relevance to several topics of recent interest in physics,  in Sec.~\ref{sec:discussion}.  The Appendices include technical details of our derivation and several side results.
}


\section{Preliminaries}
Here, we formally introduce several basic concepts and tools that will play key roles in this work.
\label{sec:preliminary}
\subsection{Approximate quantum error correction and complementary channel formalism}
\label{sec:approx-qec}

The general procedure of QEC is to first encode the logical quantum system we wish to protect by some quantum code living in a larger physical system subject to noise and errors, and perform a decoding operation with the aim of recovering the original logical information. We denote the logical and physical systems (Hilbert spaces), respectively, by $L$ and $S$.  Let $\cE^{L\to S}$ denote an encoding map that defines a code, $\cD^{S\to L}$ denote a decoding map, and $\cN^S$ denote the noise channel\footnote{We may omit the superscripts that label the associated systems when there is no ambiguity.}. 

In this work, we consider {approximate QEC}, where the code does not necessarily enable perfect recovery of the logical information but may still be useful in broad scenarios. The performance of such an approximate code can be intuitively quantified by the ``distance'' between the input and output logical states. 
Here we mainly use a well-behaved distance measure called the \emph{purified distance}. Specifically, the purified distance between two quantum states $\rho$ and $\sigma$ is defined as 
\begin{equation}
    P(\rho,\sigma) := \sqrt{1-F(\rho,\sigma)^2},
\end{equation}
where $F$ is the \emph{fidelity} defined by
\begin{equation}\label{fidelity}
    F(\rho,\sigma) := \|\sqrt{\rho}\sqrt{\sigma}\|_1 = \Tr\sqrt{\sqrt{\rho}\sigma\sqrt{\rho}}.
\end{equation}
Note that the term ``fidelity'' sometimes means $F(\rho,\sigma)^2$ (Uhlmann fidelity) in the literature, but we shall stick to the definition of Eq.~(\ref{fidelity}) in this work. 
It is known \cite[Section II]{tomamichel2010} that the purified distance satisfies the triangle inequality
\begin{equation}
    P(\rho,\sigma) \le P(\rho,{\tau}) + P({\tau},\sigma)
\end{equation}
for any state {$\tau$}. Furthermore, it satisfies the following relation with the 1-norm distance \cite{nielsen2011}:
\begin{equation}
    \frac{1}{2}\|\rho-\sigma\|_1  \le P(\rho,\sigma) \le \sqrt{2\|\rho-\sigma\|_1}. \label{eq:distances}
\end{equation}
We consider the following two most important types of characterizations of the overall performance of the  code.
The first one uses the Choi isomorphism, which considers a maximally entangled state as the input and essentially characterizes the average-case behavior of different logical states. 
More explicitly, let $R$ be a reference system with the same Hilbert space dimension as $L$, and we define the \linghang{Choi fidelity and} 
\zw{\emph{Choi error}}
of the code given by $\cE$ as
\begin{align}
   F_\choi & := \max_{\cD} F(\hat\phi^{LR}, [(\cD \circ \cN \circ \cE)^L \otimes \id^R](\hat\phi^{LR})),\\
   \epsilon_\choi& :=\sqrt{1-F_\choi^2},
\end{align}
where $\hat\phi^{LR}$ is the maximally entangled state between $L$ and $R$,
\[
    \hat\phi^{LR} = |\hat\phi\>\<\hat\phi|^{LR},\quad |\hat\phi\>^{LR} := \frac{1}{\sqrt{d_L}}\sum_{i=0}^{d_L-1} |i\>^L|i\>^R,
\]
and $d_L$ is the Hilbert space dimension of $L$. 
\zwmerge{As mentioned, the Choi error can be regarded an average-case measure since $
\epsilon_\choi = \sqrt{\frac{d_L+1}{d_L}}\epsilon_A$ \cite{gilchrist2005distance,1998quant.ph..7091H} where $\epsilon_A:= \sqrt{\int d\psi^L P^2(\psi^L,(\cD \circ \cN \circ \cE) \psi^L)}$
with the integral over the uniform Haar measure ($\psi^L$ is some pure logical state).  
In particular, as $d_L$ increases, $\epsilon_\choi$ and $\epsilon_A$ approach the same value.}
The second one is based on considering the worst-case behavior over all input states, leading to the definitions of the worst-case fidelity and \zw{\emph{worst-case error}}: 
\begin{align}
    F_\worst &:=  \max_{\cD}\min_{R,\rho^{LR}} F(\rho^{LR}, [(\cD \circ \cN \circ \cE)^L \otimes \id^R](\rho^{LR})),\\
    \epsilon_\worst &:=\sqrt{1-F_\worst^2}.
\end{align}
\linghang{Note that the minimization runs over all reference systems $R$ and all input states $\rho^{LR}$.}

The code errors $\epsilon_\choi$ and $\epsilon_\worst$ can be characterized using the formalism of complementary channels \cite{beny2010}. 
It is always possible to view $(\cN \circ \cE)^{L \to S}$ as a unitary mapping {from} $L$ to the {joint system of the physical system $S$ and an environment $E$}, followed by a partial trace over $E$.  The complementary channel of \zwnn{$(\cN \circ \cE)^{L \to S}$, denoted by $(\widehat{\mathcal{N}\circ\mathcal{E}})^{L \to E}$, is given by tracing out $S$ and outputting the state left in the environment $E$. Intuitively, an encoding is good if we cannot learn much about the input from the environment, meaning that there is not much information leaked into the environment.} To be more precise, we have
\begin{align}
    \epsilon_\choi = \min_{\zeta}P(&({\widehat{\mathcal{N}\circ\mathcal{E}}}^{L\to E}\otimes \id^R)(|\hat\phi\>\<\hat\phi|^{LR}), \nonumber \\
    &(\mathcal{T}_\zeta^{L\to E}\otimes \id^R)(|\hat\phi\>\<\hat\phi|^{LR})), 
    \\
     \epsilon_\worst = \min_{\zeta}\max_{R,|\psi\>^{LR}} P(&(\widehat{\mathcal{N}\circ\mathcal{E}}^{L\to E}\otimes \id^R)(|\psi\>\<\psi|^{LR}), \nonumber \\
     &(\mathcal{T}_\zeta^{L\to E}\otimes \id^R)(|\psi\>\<\psi|^{LR})), \label{eq:complementary}
\end{align}
where $\mathcal{T}_\zeta$ is \zwnn{a} constant channel
\[
    \mathcal{T}_\zeta^{L\to E}(\rho^L) = \Tr[\rho^L]\zeta^E
\]
\zwnn{which  outputs state $\zeta$.}
Note that it is always possible to assume the dimension of $R$ is smaller or equal to that of $L$, because the subspace perpendicular to the support of 
\zwmerge{the reduced state on $R$} does not \zwmerge{contribute to the distances}. A property of the constant channel is that
\begin{equation}
    (\mathcal{T}_\zeta^{L\to E}\otimes \id^R)(\rho^{LR}) = \zeta^E \otimes \Tr_L[\rho^{LR}] \label{eq:const-channel},
\end{equation}
which is useful for our calculations later.

\subsection{Covariant codes}\label{sec:covcode}

Let $G$ be a Lie group, and let $g\to U^S(g)$ and $g\to U^L(g)$ be representations of $G$ representing the symmetry transformations on the physical and logical Hilbert spaces respectively. We say a code is covariant with respect to $G$ if the encoding channel $\cE$ commutes with the representations, i.e.,
\begin{equation}
    \cE(U^L(g)\rho U^L(g)^\dagger) = U^S(g) \cE(\rho) U^S(g)^\dagger
\end{equation}
for all $g \in G$ and state $\rho$. \zw{A standard scenario (consider the Eastin--Knill theorem for local errors) is when  $U^S(g)$ takes the transversal (tensor product) form}
\begin{equation}
    U^S(g) = U_1(g) \otimes U_2(g) \otimes \cdots \otimes U_n(g),
\end{equation}
where $U_i(g)$ acts on the $i$-th \zw{physical subsystem} (transversal part).   \zwnn{For example, for $G = U(1)$, the symmetry transformations on the logical and physical systems are \begin{equation}
    U^L = e^{-i\theta T^L}, \quad U^S = e^{-i\theta T^S },
\end{equation}
respectively generated by charge operators (Hamiltonians) $T^L$ and $T^S$,
} 
where the tensor product structure of {$U^S$} dictates that $T^S$ takes the \zwnn{1-local} form
\begin{equation}
    T^S = \sum_{i=1}^n (T^S)_i,
\end{equation}
where $(T^S)_i$ \zwnn{is only supported on the} $i$-th qubit.

\subsection{Conditional quantum min-entropy}
\label{subsec:min-entropy}
For a bipartite quantum state $\rho^{PQ}$, the conditional min-entropy (conditioned on $Q$) is defined as
\begin{equation}
    H_{\min}(P|Q)_\rho := \sup_{\sigma\ge 0, \Tr\sigma=1}\sup\{\lambda\in \R| 2^{-\lambda}I^P\otimes \sigma^Q \ge \rho^{PQ}\}.
\end{equation}
For pure state $\psi  = |\psi\>\<\psi|^{PQ}$, there is a simple formula for $H_{\min}(P|Q)_\psi$. Let the Schmidt coefficients of $|\psi\>$ be $\alpha_1,\cdots, \alpha_D$, then the conditional min-entropy is given by
\begin{equation}
    H_{\min}(P|Q)_\psi = -2\log(\alpha_1 + \cdots + \alpha_D). \label{eq:min-max}
\end{equation}
To see this, note that \cite{konig2009} 
for any tripartite pure state $\rho$ on $X$, $Y$, and $Z$,
\begin{equation}
    H_{\min}(X|Y)_\rho + H_{\max}(X|Z)_\rho = 0,
\end{equation}
where the conditional max-entropy $H_{\max}$ is defined as
\begin{equation}
    H_{\max}(X|Z)_\rho \zwnn{:=} \sup_{\sigma^Z} \log F(\rho^{XZ},I^X \otimes \sigma^Z)^2.
\end{equation}
Note that $I^X \otimes \sigma^Z$ is not a normalized quantum state and $F(\rho^{XZ},I^X \otimes \sigma^Z)$ should be interpreted as $d_X F\left(\rho^{XZ},\frac{I^X}{d_X} \otimes \sigma^Z\right)$.

In our case, the state $|\psi\>$ is pure on $P$ and $Q$, so we can choose the third register $R$ to be a trivial system and therefore \zwnn{obtain}
\begin{align}
    H_{\min}(P|Q)_\psi =& -H_{\max}(P|R)_\psi \nonumber\\
    =& -2\log\|\sqrt{\psi^P}\sqrt{I^P}\|_1 \nonumber\\
    =& -2\log(\Tr[\sqrt{\psi^P}]) \nonumber\\
    =& -2\log(\alpha_1 + \cdots + \alpha_D). \label{eq:pure-min}
\end{align}

\subsection{Decoupling and partial decoupling}
\label{sec:partial-dec}

The (one-shot) decoupling theorem  \cite{dupuis2014} {characterizes the degree to which a system is decoupled from the environment under certain channels in terms of (suitable variants of) conditional min-entropies.} 
It can actually be viewed as a concentration-of-measure type bound
where the randomness comes from a Haar-random unitary acting on the system.  To be more precise, for any bipartite state $\rho^{AR}$ and quantum channel $\tau^{A\to E}$,  
the decoupling theorem gives an upper bound for the following quantity:
\begin{equation}
    \E_{U^A \sim \text{Haar}}\left[\left\|\mathcal{T}^{A \rightarrow E} \circ \mathcal{U}^{A}\left(\Psi^{A R}\right)-\mathcal{T}^{A \rightarrow E}\left(\Psi_\avg^{A R}\right)\right\|_{1} \right],
\end{equation}
where $\mathcal{U}^{A}$ is the superoperator defined as
\begin{equation}
    \mathcal{U}^{A}\left(\Psi^{A R}\right) := U^A \rho^{AR} U^{A\dagger}
\end{equation}
and \zwnn{$\rho^{AR}_{\avg}$ is an \zwnnn{average  state} given by}
\begin{equation}
    \quad \rho^{AR}_{\avg} := \E_{U^A \sim \text{Haar}} U^A \rho^{AR} U^{A\dagger}.
\end{equation}

A generalization of decoupling called \emph{partial decoupling} that is useful for our purpose was studied in Ref.~\cite{wakakuwa2019}, where the unitary $U^A$ \zwnew{exhibits further structures}. 
More specifically, {we assume that} the Hilbert space of $A$ takes the form of a direct-sum-product decomposition
\[
    \H^A = \bigoplus_{j=1}^J \H_j^{A_l} \otimes \H_j^{A_r},
\]
and $U^A$ satisfies
\[
    U^A = \bigoplus_{j=1}^J  I_j^{A_l} \otimes U_j^{A_r}
\]
where $U_j^{A_r}$ is Haar-random on $\H_j^{A_r}$. The distribution of such $U^A$ \zwnew{is denoted by} $H_\times$. Let $l_j$ and $r_j$ be the dimensions of $\H_j^{A_l}$ and $\H_j^{A_r}$, respectively. The (non-smoothed) partial decoupling theorem \cite[Eq.~(84)]{wakakuwa2019} states that

\begin{align}
    &\mathbb{E}_{U \sim \zwnew{H_\times}}\left[\left\|\mathcal{T}^{A \rightarrow E} \circ \mathcal{U}^{A}\left(\Psi^{A R}\right)-\mathcal{T}^{A \rightarrow E}\left(\Psi_\avg^{A R}\right)\right\|_{1}\right] \nonumber\\
    \leq& 2^{-\frac{1}{2} H_{\min }(A^*|RE)_{\Lambda(\Psi, \mathcal{T})}}.
\end{align}
Here, the state $\Lambda(\Psi, \mathcal{T})$ is defined as
\begin{equation}\label{eq:lambda1}
    \Lambda(\Psi, \mathcal{T}) := \Xi(\Psi^{AR} \otimes \tau^{\bar A E})\Xi^\dagger,
\end{equation}
where $\tau^{\bar A E}$ is the Choi state of $\mathcal{T}$ and the operator \zwnn{$\Xi^{A\bar A \to A^*}$ is a map from \lh{$\H^{A} \otimes  \H^{\bar A}$ to $\H^{A^*} := \bigoplus_{j=1}^J \H_j^{A_r} \otimes \H_j^{\bar A_r}$} given by}
\begin{equation}\label{eq:lambda_F}
    \Xi^{A \bar{A} \rightarrow A^{*}} := \bigoplus_{j=1}^J \sqrt{\frac{d_A l_j}{r_j}}\<\Phi_j^l|^{A_l\bar A_l}\left(\Pi_{j}^{A} \otimes \Pi_{j}^{\bar{A}}\right)
\end{equation}
with $|\Phi_j^l\>$ being the maximally entangled state, and $\Pi_{j}^{A}$ being the projector onto $\H_j^{A_l} \otimes \H_j^{A_r}$. \zwnn{Also, $ \Psi_{\avg}^{A R}$ is again the \zwnnn{average state} defined by}
\begin{equation}
    \Psi_{\avg}^{A R} := \E_{U\sim H_\times} U \Psi^{A R} U^\dagger,
\end{equation}
\zwnn{which \zwnnn{takes the form}
\begin{equation}
    \Psi_{\avg}^{A R} = \bigoplus_{j=1}^J \Psi_{jj}^{A_l R} \otimes \frac{I_j^{A_r}}{r_j} \label{eq:avg-state}
\end{equation}
where
     $\Psi_{jj}^{A_l R} := \Tr_{A_r}[\Pi_j^A \Psi^{AR} \Pi_j^A].$
}

\section{$U(1)$ symmetry}
\label{sec:u1}
\zwnew{In this section, we consider the most basic $U(1)$ symmetry case, \zwnn{which is Abelian and  corresponds to a single conservation law}. }

\subsection{\zwnew{$U(1)$-covariant codes from charge-conserving unitaries}}
\label{sec:u1-construction}
\zwnew{We first formally define our code construction which will be analyzed.} \zwnn{Here we consider quantum codes that are covariant with respect to the $U(1)$ symmetry, or namely conserve the $U(1)$ charge, in the sense introduced in Sec.~\ref{sec:covcode}. That is, let $\cE$ be the encoding channel, and let the $U(1)$ group action $e^{i\theta} \in U(1)$ be represented as $e^{i\theta} \to e^{-i\theta T^L}$ and $e^{i\theta} \to e^{-i\theta T^S}$ generated by a logical charge operator $T^L$ and a transversal physical operator $T^S$ on the logical and physical Hilbert spaces, respectively. 
Then the covariance condition requires that $\cE^{L\rightarrow S}\circ \mathcal{U}^L = \mathcal{U}^S \circ\cE^{L\rightarrow S}$ where $\mathcal{U}^{L(S)} := e^{-i\theta T^{L(S)}}(\cdot) e^{i\theta T^{L(S)}}$, namely that the encoding map commutes with group actions.}
\zwnn{Without loss of generality, we consider the Hamming weight operator as the charge operator. On $m$ qubits, it is defined by}
\[
    Q^{(m)} := \sum_{i=1}^m \frac{I - Z_i}{2},
\]
where $Z = |0\>\<0| - |1\>\<1|$ is the Pauli-Z operator on a single qubit, and $Z_i$ is the operator $Z$ acting on qubit $i$. 
\zwnn{We consider codes that encode $k$ logical qubits in $n$ physical qubits, so that the logical and physical charge operators  are, respectively, $T^L = Q^{(k)}$  and $T^S = Q^{(n)}$.} 

Our code construction relies on $n$-qubit unitaries that commute with $Q^{(n)}$, i.e.,~conserve the Hamming weight. Such unitaries are block diagonal with respect to the eigenspaces of $Q^{(n)}$\zwmerge{, forming a compact subgroup of the unitary group} \lhn{denoted by $\cU_{\times, U(1)}$}. 
Let \zwnew{$H_{\times,{U(1)}}$} denote the Haar measure on the group \lhn{$\cU_{\times,U(1)}$.}
\begin{defi}
\zwnn{A code is called} a \emph{\zwnew{$(n,k;\alpha)$-$U(1)$ code}} if it encodes $k$ logical qubits in $n$ physical qubits by first appending an $(n-k)$-qubit state $|\psi_\alpha\>$ that \zwnew{has Hamming weight $\alpha$, i.e.,~satisfies}
\[
    Q^{(n-k)} |\psi_\alpha\> = \alpha |\psi_\alpha\>,
\]
and then applying a unitary $U$ on the joint $n$-qubit system that commutes with $Q^{(n)}$.  In particular, a \emph{\zwnew{$(n,k;\alpha)$-$U(1)$ random code}} is given by $U$ drawn from \zwnew{$H_{\times,{U(1)}}$}.
\end{defi}
It is straightforward to verify that such codes indeed satisfy the covariance condition.
\begin{prop}
\label{thm:cov}
 $(n,k;\alpha)$-$U(1)$ codes are covariant with respect to the $U(1)$ group \zwnew{ generated by charge operators $T^L=Q^{(k)}$ and  $T^S = Q^{(n)}$ on the logical and physical systems respectively}. 
This property holds for the \zwnew{$(n,k;\alpha)$-$U(1)$ random code}.
\end{prop}
\begin{proof}
Since the $n$-qubit unitary $U$ commutes with $Q^{(n)}$, it commutes with $e^{iQ^{(n)}\theta}$ for all $\theta$. Then for any $k$-qubit logical state $\rho$ we have
\begin{align}
    &e^{i Q^{(n)} \theta} U(\rho \otimes |\psi_\alpha\>\<\psi_\alpha|)U^\dagger e^{-iQ^{(n)}\theta} \nonumber\\
    =& U e^{i Q^{(n)} \theta} (\rho \otimes |\psi_\alpha\>\<\psi_\alpha|) e^{-iQ^{(n)}\theta} U^\dagger \nonumber\\
    =& U  (e^{i Q^{(k)} \theta}\rho e^{-i Q^{(k)} \theta} \otimes e^{i Q^{(n-k)} \theta}|\psi_\alpha\>\<\psi_\alpha|e^{-i Q^{(n-k)} \theta})  U^\dagger \nonumber\\
    =& U  (e^{i Q^{(k)} \theta}\rho e^{-i Q^{(k)} \theta} \otimes |\psi_\alpha\>\<\psi_\alpha|)  U^\dagger,
\end{align}
which means the encoding map $\cE$ satisfies the covariance condition
\[
    e^{i Q^{(n)} \theta} \cE(\rho) e^{-i Q^{(n)} \theta} = \cE(e^{i Q^{(k)} \theta}\rho e^{-i Q^{(k)} \theta}).
\]
\end{proof}

\zwnew{We emphasize that our \zwnew{$(n,k;\alpha)$-$U(1)$ random code} not only represents a randomized construction of $U(1)$-covariant codes, but also reveals the average or typical properties of all charge-conserving unitaries due to the Haar measure.
}

\subsection{Performance of random $U(1)$-covariant codes}
\label{sec:u1-error}

We now explicitly analyze both \zwnn{the} Choi error and the worst-case error of the \zwnew{$(n,k;\alpha)$-$U(1)$ random code}  in terms of purified distance, against \zwnn{the} erasure of $t$ qubits.  
\zwnew{Here, for simplicity of exposition, we fix the set of erased qubits, but note that the results still hold when the erased qubits are picked randomly as will be discussed in Sec.~\ref{sec:discussion}.}
\zwnn{Note that the complementary channel of the noise channel, namely \zwnn{the} erasure of $t$ qubits, is simply a partial trace over the other $n-t$ unaffected qubits, which we denote by $\Tr_{n-t}[\cdot]$.}

\subsubsection{Choi error}
\label{sec:choi}


\begin{theo}
\label{thm:choi}
\linghang{
In the large $n$ limit, \zwnn{suppose} $k$ and $t$ satisfy $k^2t^2 = o(n)$ and $\alpha=a n$ \zwnn{for some constant $a \in (0,1)$}, then the \zw{expected Choi error of the $(n,k;\alpha)$-$U(1)$ random code} satisfies}
\begin{equation}\label{eq:choi_expectation}
    \E \epsilon_\choi \le  \frac{\sqrt{tk}}{4n\sqrt{a(1-a)}}\left(1+O\left(\frac{k^2t^2}{n}\right)\right).
\end{equation}
Furthermore, the probability that \zw{the Choi error of a $(n,k;\alpha)$-$U(1)$ code (with respect to $H_{\times,U(1)}$)}  violates the  inequality above is exponentially small in $n$, i.e.,
\begin{equation} \label{eq:choi_prob}
\Pr\left[\epsilon_\choi >  \frac{\sqrt{tk}}{4n\sqrt{a(1-a)}}\left(1+O\left(\frac{k^2t^2}{n}\right)\right)\right] = e^{-\Omega(n)}.
\end{equation}
\end{theo}

\begin{proof}
\zwnew{To analyze the error, we employ the complementary channel formalism \zwnn{introduced in Sec.~\ref{sec:approx-qec}}. }
By Eq.~\eqref{eq:complementary}, for a specific choice of $U$ in our $(n,k;\alpha)$-$U(1)$ code construction, the Choi error of the corresponding code satisfies
\begin{align}
    &\epsilon_\choi  \nonumber\\
    =& \min_\zeta P\left(\Tr_{n-t}[U(|\hat\phi\>\<\hat\phi|\otimes |\psi_\alpha\>\<\psi_\alpha|)U^\dagger],\frac{I}{2^k}\otimes \zeta\right) \nonumber\\
    \le& P\left(\Tr_{n-t}[U(|\hat\phi\>\<\hat\phi|\otimes |\psi_\alpha\>\<\psi_\alpha|)U^\dagger],\Tr_{n-t}\Phi_\text{avg}\right) \nonumber\\
    &+\min_\zeta P\left(\Tr_{n-t}\Phi_\text{avg},\frac{I}{2^k}\otimes \zeta\right) \nonumber \\
    \le& \sqrt{2\left\|\Tr_{n-t}[U(|\hat\phi\>\<\hat\phi|\otimes |\psi_\alpha\>\<\psi_\alpha|)U^\dagger]-\Tr_{n-t}\Phi_\text{avg}\right\|_1} \nonumber\\
    &+\min_\zeta P\left(\Tr_{n-t}\Phi_\text{avg},\frac{I}{2^k}\otimes \zeta\right), \label{eq:bound}
\end{align}
where \lh{$|\psi_\alpha\>$ has eigenvalue $\alpha$ and could be taken as $|\psi_\alpha\> = |1 \cdots 10 \cdots 0\>$ ($\alpha$ 1's and $n-\alpha$ 0's) without loss of generality,} \zwnew{and $\Phi_\text{avg}$ is \zwnnn{the average physical state} given by}
\begin{equation}
    \Phi_\text{avg} = \E_{U\sim H_{\times,U(1)}} U(|\hat\phi\>\<\hat\phi|\otimes |\psi_\alpha\>\<\psi_\alpha|)U^\dagger.
\end{equation}
Eq.~\eqref{eq:const-channel} is used in the first line, \zw{and the second line follows from the triangle inequality of $P$}.

When averaging over $U$ sampled from $H_{\times,U(1)}$, the first term in Eq.~\eqref{eq:bound} can be bounded using the partial decoupling theorem:
\begin{align}
    &\E_{U} \sqrt{2\left\|\Tr_{n-t}[U(|\hat\phi\>\<\hat\phi|\otimes |\psi_\alpha\>\<\psi_\alpha|)U^\dagger]-\Tr_{n-t}\Phi_\text{avg}\right\|_1} \nonumber\\
    \le& \sqrt{2\E_{U}\left\|\Tr_{n-t}[U(|\hat\phi\>\<\hat\phi|\otimes |\psi_\alpha\>\<\psi_\alpha|)U^\dagger]-\Tr_{n-t}\Phi_\text{avg}\right\|_1} \nonumber \\
    \le & \sqrt{2} \times 2^{-\frac{1}{4}H_{\min}(A^*|RE)_{\Lambda(\Psi, \cT)}}.
\end{align}
\zwnn{Here, $\Psi = |\hat\phi\>\<\hat\phi|\otimes |\psi_\alpha\>\<\psi_\alpha|$  is the initial state  on $AR$, and $\cT^{A\rightarrow E}$ is the complementary erasure $\Tr_{n-t}$. Then $\Lambda(\Psi, \cT)$ is  a tripartite state on $A^*RE$ as defined in Eq.~(\ref{eq:lambda1}), obtained by applying the operator $\Xi^{A\bar A\to A^*}$ defined in Eq.~(\ref{eq:lambda_F}) to the tensor product of $\Psi$ and the Choi state of $\cT$ (on $\bar A E$). See Sec.~\ref{sec:partial-dec} for details. Then $H_{\min}(A^*|RE)_{\Lambda(\Psi, \cT)}$ is the conditional min-entropy of this state conditioned on $RE$ (see Sec.~\ref{subsec:min-entropy}).}

We prove in Appendix~\ref{app:entropy} that
\begin{equation}\label{eq:hmin}
    H_{\min }(A^*|RE)_\zwnn{{\Lambda(\Psi, \cT)}} = \Omega(n),
\end{equation}
which implies that the expectation value of the first term is exponentially small. A simple application of Markov's inequality shows that this term is exponentially small with probability equal to 1 minus an exponentially small quantity.

The second term in Eq.~\eqref{eq:bound} is independent of $U$. Since this is a minimization, an upper bound on this term can be found by any choice of $\zeta$. We set $\zeta$ to be the $t$-qubit marginal state of $\Tr_{n-t}[\Phi_\text{avg}]$, and as detailed in Appendix~\ref{app:average-state}, we obtain
\begin{equation}
    P\left(\Tr_{n-t}\Phi_\text{avg},\frac{I}{2^k}\otimes \zeta\right) \le \frac{\sqrt{tk}}{4n\sqrt{a(1-a)}}\left(1+O\left(\frac{t^2k^2}{n}\right)\right).
\end{equation}

\zwnew{Then the claimed bounds  follows from combining the above analysis of the two terms in Eq.~\eqref{eq:bound}.}

\end{proof}

\lh{Note that in \zwnn{Appendix}~\ref{sec:alternative}, we provide an alternative proof of Eq.~\eqref{eq:hmin}.  There we actually give an exact expression for $\Lambda(\Psi,\tau)$, yielding a lower bound as well as an upper bound for the \zwnn{conditional} min-entropy. The discussion may be of independent interest.}

\subsubsection{Worst-case error}
\label{sec:worst-case}
\zwnn{We need the following lemma, which} gives a lower bound of \zwnew{the worst-case purified distance}  for a fixed code.
\begin{lemma}[{\cite[Prop.~4]{faist2019}}]
\label{lem:worst-bound}
Given encoding channel $\mathcal{E}$ and noise channel $\mathcal{N}$, let $\widehat{\mathcal{N}\circ\mathcal{E}}$ be a complementary channel of $\mathcal{N}\circ\mathcal{E}$. \zwnn{For a fixed basis} of logical states $\{|x\>\}$, we define
\begin{equation}
    \rho^{x,x'} \zwnew{:=} \widehat{\mathcal{N}\circ\mathcal{E}} (|x\>\<x'|). \label{eq:rho-def}
\end{equation}
Suppose that there exists a state $\zeta$ as well as constants $\epsilon,\epsilon'>0$ such that
\begin{align}
    P(\rho^{x,x}, \zeta) \le& \epsilon, \label{eq:rhoxx}\\
    \|\rho^{x,x'}\|_1 \le & \epsilon', \quad \forall x \not= x'. \label{eq:rhoxxp}
\end{align}
Then, the code given by $\mathcal{E}$  \zwnn{satisfies}
\begin{equation}
    \epsilon_\worst \le \epsilon+ d_L \sqrt{\epsilon'}, \label{eq:worst-bound}
\end{equation}
where $d_L$ is the dimension of the logical system.

If one of several noise channels is applied at random but it is known which one occurred, then Eq.~\eqref{eq:worst-bound} holds for the overall noise channel if the assumptions above are satisfied for each individual noise channel.
\end{lemma}
For our $(n,k;\alpha)$-$U(1)$ code construction, 
\begin{equation}
 {\rho^{x,x'}} = \Tr_{n-t}[U(|x\>\<x'| \otimes |\psi_\alpha\>\<\psi_\alpha|)U^\dagger].
\end{equation}
Note that the lemma above applies to a fixed encoding $\mathcal{E}$. To generalize this theorem to our randomized construction, we define $\rho^{x,x'}_\avg$ as $\rho^{x,x'}$ in Eq.~\eqref{eq:rho-def} averaged over the random unitary in $\mathcal{E}$. Then using the following lemma, we can obtain bounds on the worst-case error. 
\begin{lemma}
Consider the large $n$ limit. 
If the \zwnnn{average physical states} $\rho^{x,x}_\avg$ satisfy $P(\rho^{x,x}_\avg, \zeta) \le \epsilon$  for some fixed state $\zeta$ independent of $x$, then with probability at least $1-p_1-p_2$, \lh{our $(n,k;\alpha)$-$U(1)$ random code} satisfy
\begin{equation}
    \epsilon_\worst \le \epsilon + \delta + 2^k \sqrt{\delta'},
\end{equation}
\zwnn{where $p_1$ and $p_2$ are given by}
\begin{align}
    \log p_1 =& k-\frac{n}{4}\min\left\{H\left(\frac{\alpha}{n}\right),H\left(\frac{\alpha+k}{n}\right)\right\} \nonumber\\
    &+ \frac{t}{2}+\log\frac{1}{\delta} + O(\log n),\\
    \log p_2 =& 2k-\frac{n}{2} \min\left\{H\left(\frac{\alpha}{n}\right),H\left(\frac{\alpha+k}{n}\right)\right\} \nonumber\\
    &+ t + \log\frac{1}{\delta'} + O(\log n). \label{eq:p1p2}
\end{align}
\label{lem:worst-randomized}
\end{lemma}
\begin{proof}
We use the partial decoupling theorem to  upper bound  the average distance between $\rho^{x,x'}_\avg$ and $\rho^{x,x'}$. Then by Markov inequality \zwnn{we can bound the} probability that $\rho^{x,x'}$ behaves much worse than $\rho^{x,x'}_\avg$. \zwnn{Then we can} show the code has good performance with high probability using a union bound.

Given that there exists $\zeta$ such that $P(\rho^{x,x}_\avg, \zeta) \le \epsilon$ for all $x$, we have
\begin{align}
    \E_U P(\rho^{x,x}, \zeta) \le& P(\rho^{x,x}_\avg, \zeta) + \E_U P(\rho^{x,x}, \rho^{x,x}_\avg) \nonumber\\
    \le& \epsilon + \E_U \sqrt{2\|\rho^{x,x}- \rho^{x,x}_\avg\|_1} \nonumber \\
    \le& \epsilon +  \sqrt{2\E_U \left[\|\rho^{x,x}- \rho^{x,x}_\avg\|_1\right]} \nonumber \\
    \le& \epsilon + \sqrt{2} \times 2^{-\frac{1}{4}H_{\min}^{x}},
\end{align}
where $H_{\min}^{x}\zwnn{\equiv}H_{\min }(A^*|RE)_{\Lambda(\zwnn{\Psi^x}, \mathcal{T})}$ \zwnn{for which the initial state is $\Psi^x \equiv |x\>\<x|$ and $\cT$ is $\Tr_{n-t}$ (detailed definitions introduced in Thm.~\ref{thm:choi} and Sec.~\ref{sec:partial-dec}).} 
\zwnn{As shown in} Appendix~\ref{app:entropy},
\begin{align}
    H_{\min}^x &= nH\left(\frac{|x|+\alpha}{n}\right)-2t+\log(n) \nonumber\\
    &\ge n \min\left\{H\left(\frac{\alpha}{n}\right),H\left(\frac{\alpha+k}{n}\right)\right\} - 2t + O(\log n),
\end{align}
which \zwnn{implies} that for each $x$,
\begin{align}
    & \log \Pr[P(\rho^{x,x},\zeta) \ge \epsilon+\delta] \nonumber \\
    \le & \, \frac{1}{2}-\frac{1}{4}H_{\min}^x + \log\frac{1}{\delta} \nonumber \\
    \le & -\frac{n}{4}\min\left\{H\left(\frac{\alpha}{n}\right),H\left(\frac{\alpha+k}{n}\right)\right\} + \frac{t}{2}+\log\frac{1}{\delta} + O(\log n). \label{eq:prob-x}
\end{align}

When $x\not=x'$, it is easy to see that $\rho^{x,x'}_\avg=\E_{U\sim H_{\times,U(1)}}\Tr_{n-t}[U(|x\>\<x'| \otimes |\psi_\alpha\>\<\psi_\alpha|)U^\dagger]=0$. Since the partial decoupling theorem applies only to subnormalized states, we \zwnn{need to express $|x\>\<x'|$ as}
\begin{align}
    |x\>\<x'| =& \frac{1}{2}|\mu_{x,x'}^+\>\<\mu_{x,x'}^+| - \frac{1}{2}|\mu_{x,x'}^-\>\<\mu_{x,x'}^-| \nonumber \\
    &+ \frac{i}{2}|\nu_{x,x'}^+\>\<\nu_{x,x'}^+| - \frac{i}{2}|\nu_{x,x'}^-\>\<\nu_{x,x'}^-|,
\end{align}
where
\begin{align}
    |\mu^\pm_{x,x'}\> &\zwnn{\equiv} \frac{1}{\sqrt 2}(|x\> \pm |x'\>), \\
    |\nu^\pm_{x,x'}\> &\zwnn{\equiv} \frac{1}{\sqrt 2}(|x\> \pm i|x'\>).
\end{align}
Then we can apply the partial decoupling theorem to the states $|\mu^\pm_{x,x'}\>$ and $|\nu^\pm_{x,x'}\>$ and \zwnn{obtain}
\begin{align}
    \E_U \|\rho^{x,x'}\|_1 \le& \frac{1}{2}\times2^{-\frac{1}{2}H_{\min}^{x,x',\mu^+}}+\frac{1}{2}\times2^{-\frac{1}{2}H_{\min}^{x,x',\mu^-}}\nonumber\\
    &+\frac{1}{2}\times2^{-\frac{1}{2}H_{\min}^{x,x',\nu^+}}+\frac{1}{2}\times2^{-\frac{1}{2}H_{\min}^{x,x',\nu^-}},
\end{align}
where $H_{\min}^{x,x',\mu^\pm}$ and $H_{\min}^{x,x',\nu^\pm}$ are $H_{\min }(A^*|RE)_{\Lambda(\Psi, \mathcal{T})}$  with the initial state \zwnn{$\Psi$ taken  to be $|\mu^\pm_{x,x'}\>\<\mu^\pm_{x,x'}|$ and $|\nu^\pm_{x,x'}\>\<\nu^\pm_{x,x'}|$, respectively, and $\cT$ being $\Tr_{n-t}$}.
\zwnn{As shown in} Appendix~\ref{app:entropy}, we have
\begin{equation}
    H_{\min}^{x,x'} \ge n \min\left\{H\left(\frac{\alpha}{n}\right),H\left(\frac{\alpha+k}{n}\right)\right\} - 2t + O(\log n),
\end{equation}
where $H_{\min}^{x,x'}$ could be any one among $H_{\min}^{x,x',\mu^\pm}$ and $H_{\min}^{x,x',\nu^\pm}$. \zwnn{By Markov inequality we obtain}
\begin{align}
    \log \Pr\left[\|\rho^{x,x'}\|_1 \ge \delta'\right] \le& -\frac{n}{2} \min\left\{H\left(\frac{\alpha}{n}\right),H\left(\frac{\alpha+k}{n}\right)\right\} \nonumber\\
    &+ t + \log\frac{1}{\delta'} + O(\log n). \label{eq:prob-xx}
\end{align}
Now we can apply \zwnn{the} union bound and take \zwnn{the} sum of Eq.~\eqref{eq:prob-x} over all $x$ and Eq.~\eqref{eq:prob-xx} over all $x$ and $x'$. This means that $P(\rho^{x,x},\zeta) \le \epsilon+\delta$ and $\|\rho^{x,x'}\|_1 \le \delta'$ are satisfied for all $x$ and $x'$ with probability at least $1-p_1-p_2$ where $p_1$ and $p_2$ are defined in Eq.~\eqref{eq:p1p2}. By Lemma~\ref{lem:worst-bound}, the code satisfies
\begin{equation}
    \epsilon_\worst \le \epsilon + \delta + 2^k \sqrt{\delta'}
\end{equation}
\zwnn{with probability at least $1-p_1-p_2$.}
\end{proof}
\begin{theo}
\label{thm:worst}
\linghang{In the large $n$ limit}, 
\zwnn{suppose} $k$ and $t$ satisfy $kt^2=o(n)$ and $\alpha=an$ \zwnn{for some constant $a \in (0,1)$}, then \zw{the expected worst-case error of the $(n,k;\alpha)$-$U(1)$ random code} satisfies \begin{equation}
    \E \epsilon_\worst \le \frac{k\sqrt{t}}{4n\sqrt{a(1-a)}}\left(1+O\left(\frac{kt^2}{n}\right)\right). \label{eq:worst-bound2}
\end{equation}
Furthermore, the probability 
\zw{\zw{that the worst-case error of a $(n,k;\alpha)$-$U(1)$ code (with respect to $H_{\times,U(1)}$)}  violates the  inequality above is exponentially small in $n$, i.e.,}
\begin{equation}
    \Pr\left[\epsilon_\worst > \frac{k\sqrt{t}}{4n\sqrt{a(1-a)}}\left(1+O\left(\frac{kt^2}{n}\right)\right)\right]= e^{-\Omega(n)}.
\end{equation}
\end{theo}
\begin{proof}
In Appendix~\ref{app:worst-avg}, we show that $P(\rho^{x,x}_\avg, \zeta)$ is upper bounded by the right hand side of Eq.~\eqref{eq:worst-bound2}. Now we apply Lemma~\ref{lem:worst-randomized} with properly chosen exponentially small $\delta$ and $\delta'$ so that $p_1$ and $p_2$ are also exponentially small, which shows that the code satisfies the inequality \zw{with exponentially small failure probability}. \linghang{Since $\epsilon_\worst$ is at most 1, this  implies that the expectation of $\epsilon_\worst$ satisfies the inequality as well.}
\end{proof}

\zwnn{Note that Appendix~\ref{sec:alternative} also includes an alternative derivation of the conditional min-entropies that show up in the above proof for the worst-case error, based on exact expressions of $\Lambda$, }\lhn{which may be of independent interest.}

\subsubsection{Remarks on  \zwnew{input charge and distance metric}}
\label{sec:remarks}

 \zwnn{It is interesting to note how the charge of the input ancilla state in our construction affects the code performance}.  In particular, $a=1/2$ (namely $\alpha=n/2$) gives rise to the best accuracy, and the accuracy becomes worse as one increases or decreases $a$. The intuition is that  the code resides in a subspace with Hamming weight between $\alpha$ and $\alpha+k$, \zwnn{which has the largest size and apparently the strongest entanglement that enhance the performance of the code, when $\alpha$ is around $n/2$. }

\zw{
\zwnn{Specifically, consider the scaling of the code errors.  
As we have shown, for linear $\alpha$ (constant $a$), the partial decoupling terms  (such as the first term of Eq.~(\ref{eq:bound})) are exponentially small in $n$, so the remaining symmetry terms (such as the second term of Eq.~(\ref{eq:bound})), which are polynomially small, are dominant in the overall code errors. Even when $\alpha$ is constant, following the calculation in Appendix~\ref{app:entropy}, we have that the conditional min-entropy is at least $\alpha\log n$, so for large $\alpha$, the partial decoupling terms are still dominated.} In particular, the error scales as $n^{-1/2}$ for constant input charge $\alpha$ and $n^{-1}$ for linear $\alpha$ (constant $a$).

\zwnn{Another interesting observation is that, for constant $\alpha$, the error bounds may behave significantly differently when using the trace or 1-norm distance instead of the purified distance in the definitions. Consider, for example, the Choi case.
According to our numerical results as shown in Fig.~\ref{fig:distance} and Table~\ref{tab:slope}, the symmetry term (the second term) in Eq.~\eqref{eq:bound} given by the purified distance scales worse for constant $\alpha$ than for linear $\alpha$. To be more precise, by a linear fitting, it can be seen that this term 
indeed scales roughly like $n^{-1/2}$ for constant $\alpha$ and $n^{-1}$ when $\alpha= O(n)$  (see Table~\ref{tab:slope}).}    On the other hand, this symmetry term as given by the trace distance always scales like $n^{-1}$.  \linghang{These numerical results are consistent with our calculation in Appendix~\ref{app:average-state}.}   \zwnew{Namely, a desirable feature of the purified distance is that it can distinguish different input charge scalings by its scaling. In fact, the cases of constant and linear input charge correspond to the two extremes in Eq.~\eqref{eq:distances}. }  }

\begin{figure}[ht]
    \includegraphics[width=0.49\textwidth]{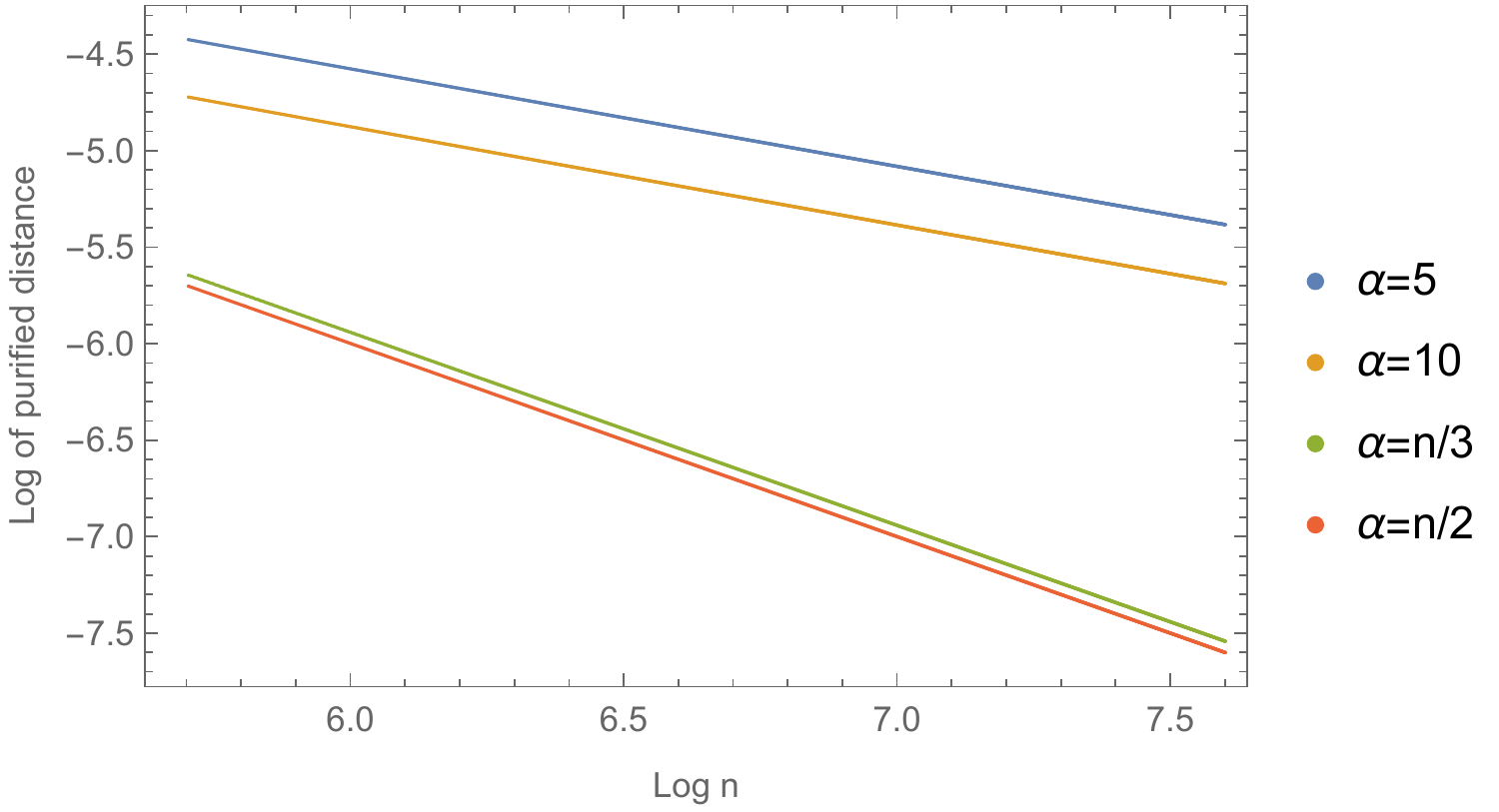}
    \includegraphics[width=0.49\textwidth]{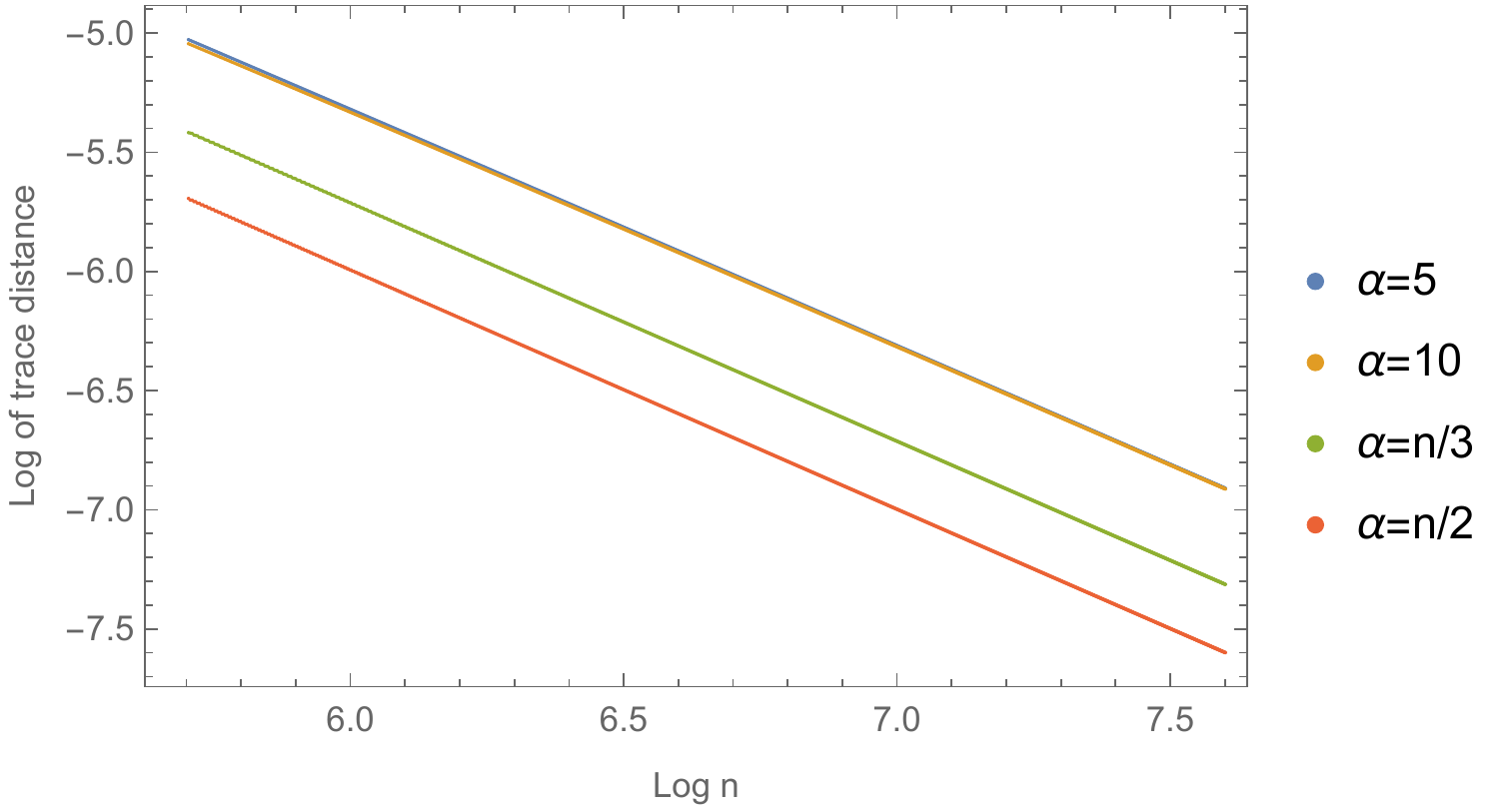}
    \caption{Log-log plots of \zw{the symmetry term in Eq.~\eqref{eq:bound} 
    as given by the \lhn{trace} distance and the purified distance for different $\alpha$ (here we set $k=t=2$).} 
    }
    \label{fig:distance}
\end{figure}
\begin{table}
\centering
\begin{tabular}{c|c|c|c|c}
    \hline
    $\alpha$ & 5 & 10 & $n/3$ & $n/2$\\
    \hline
    \lhn{Trace} distance & -0.993676 & -0.987255 & -0.999996 & -1.00313 \\
    \hline
    Purified distance & -0.504438 & -0.507633 & -1.00001 & -1.00063 \\
    \hline
\end{tabular}
\caption{\zwnn{Numerical values of} the slopes of the lines in Fig.~\ref{fig:distance} \zwnn{from linear fitting.}}
\label{tab:slope}
\end{table}

\subsection{Comparisons with fundamental limits}
\label{sec:u1-compare}
Now let us compare the performance of our $(n,k;\alpha)$-$U(1)$ codes with \zwnew{known lower bounds for $U(1)$-covariant} codes.
For simplicity, consider the $t=1$ case, namely \zwnn{the single-erasure noise channel}.
For $U(1)$ symmetry, Thm.~1 in Ref.~\cite{faist2019} indicates the following lower bounds:
\begin{align}
    \epsilon_\choi &\ge \frac{\binom{k}{\lceil k/2 \rceil}\lceil k/2 \rceil}{2^k n}, \\ \epsilon_\worst &\ge \frac{k}{2n}. \label{eq:lower-bounds}
\end{align}
Note that for large $k$, the bound on $\epsilon_\choi$ {approaches}
\begin{equation}
    \epsilon_\choi \ge \frac{1}{n}\sqrt{\frac{k}{2\pi}}.
\end{equation}
In comparison, \zwnn{according to} Thm.~\ref{thm:choi} and Thm.~\ref{thm:worst}, our $(n,k;\alpha)$-$U(1)$ random code has smallest error when $a=1/2$, in which case
\begin{align}
    \epsilon_\choi &\le \frac{\sqrt{k}}{2n}(1+O(k^2/n)),\\ 
    \epsilon_\worst &\le \frac{k}{2n}(1+O(k/n)).
\end{align}
That is, {up to leading order}, the worst-case \zwnn{error} exactly matches the lower bound, and the Choi \zwnn{error} matches the bound up to a constant factor. When $k=1$, the Choi \zwnn{error} also matches the bound in Eq.~\eqref{eq:lower-bounds} exactly.

The situation for the general $t>1$  case is as follows.
Thm.~2  in Ref.~\cite{faist2019} gives a bound for general erasure: If the physical charge operator has the form
\begin{equation}
    T^S = \sum_\eta \zwnn{T_\eta},
\end{equation}
where $T_\eta$ has support on a set of qubits $\eta$, and that $\eta$ gets erased with probability $q_\eta$, then
\begin{align}
    \epsilon_\choi &\ge \frac{\left\|T^{L}-\mu(T^{L}) I^{L}\right\|_{1} / d_{L}}{\max _{\eta}\left(\Delta T^{\eta} / q_{\eta}\right)}, \\
    \epsilon_\worst &\ge \frac{\Delta T^{L} }{2\max _{\eta}\left(\Delta T^{\eta} / q_{\eta}\right)},\label{eq:t_bound}
\end{align}
where $\Delta T^\eta$ and $\Delta T^L$ are the difference between the largest and smallest eigenvalues of $T^\eta$ and $T^L$ respectively, and $\mu(T^L)$ is the median of the eigenvalues of $T^L$.
For example, consider the following two ways to model the erasure of $t$ qubits:
\begin{enumerate}
    \item The qubits are grouped into sets of size $t$, and there are $n/t$ such sets. Each $T^\eta$ is the Hamming weight operator on this set, and $\Delta T^\eta = t, q_\eta = t/n$;
    \item $\eta$ represent all possible sets of $t$ qubits. $q_\eta=\binom{n}{t}^{-1}$, and each $T_\eta$ is $\frac{n}{t}\times \binom{n}{t}^{-1}$ the Hamming weight operator on the set, where the coefficient is chosen so that $T_\eta$ sum up to the physical charge operate $T_S$. In this case $\Delta T^\eta = n\times \binom{n}{t}^{-1}, q_\eta=\binom{n}{t}^{-1}$.
\end{enumerate}
In either case, it turns out that
\begin{equation}
    \max_\eta \left(\Delta T^{\eta} / q_{\eta}\right) = n.
\end{equation}
As a result, the lower bounds in Eq.~(\ref{eq:t_bound}) do not scale with $t$ and are expected to be loose. 
\zw{We leave potential improvement of the lower bounds in Ref.~\cite{faist2019} as well as more careful investigation into different methods in, e.g.,~Refs.~\cite{Woods2020continuousgroupsof,YangWoods,KubicaD,zlj20} for the $t>1$ case for future work.}

\section{$SU(d)$ symmetry}
\label{sec:sud}
We now proceed to discuss the more complicated case of $SU(d)$ symmetry, \zwnew{which is \zwnn{non-Abelian and particularly important for fault-tolerant}  quantum computing as it describes the entire group of unitary gates.  The high-level ideas are reminiscent to the $U(1)$ case, but more advanced representation theory techniques will be used.  Here, we first review aspects of the representation theory of $SU(d)$ and the relevant permutation group $S_n$ that play key roles in our analysis, and then present the results.  
} 

\subsection{Representation theory of $SU(d)$ and $S_n$}
\label{sec:sud-pre}
\subsubsection{Partitions and Young tableaux}
Let \lhn{$n$ be a positive integer and }$\lambda=(\lambda_1, \cdots, \lambda_m)$ be a set of non-increasing \zwnn{positive integers} \zwnn{such that} $n=\lambda_1 + \cdots + \lambda_m$. Then $\lambda$ \zwnn{describes} a \emph{partition} of $n$, denoted by $\lambda \vdash n$. Sometimes we would like to emphasize that $\lambda$ is a partition of $n$ with at most $d$ terms, denoted by $\lambda \vdash (n,d)$.

It is often helpful to think of $\lambda$ as a set of boxes \zwnn{arranged in a way that there are $\lambda_i$ boxes in the $i$-th row, represented by the so-called  \emph{Young diagrams}}. 
One can fill numbers into a Young diagram and make it a \emph{Young tableau}. There are two types of Young tableau that are relevant to our work:
\begin{enumerate}
    \item A \emph{standard Young tableau} is obtained by filling 1 to $n$ into each box of a Young diagram. Each number appears exactly once, and the numbers on each row and each column should be increasing.
    \item A \emph{semistandard Young tableau} contains numbers that could be repeated. The numbers on each row are weakly increasing, and the numbers on each column are strictly increasing.
\end{enumerate}
\zwnn{A basic example of a Young diagram and associated standard and semistandard Young tableaux}  can be found in Fig.~\Ref{fig:ytab}.
\begin{figure}[ht]
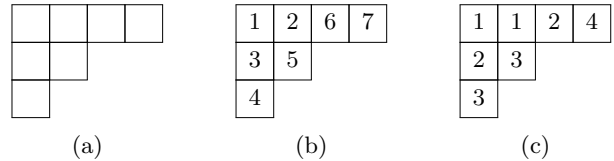

     \centering
     \begin{subfigure}[b]{0.15\textwidth}
         \centering
         \begin{ytableau}
            {} &&& \\
            {} & \\
            {}
         \end{ytableau}
         \caption{}
         \label{fig:y-diag}
     \end{subfigure}
     \hfill
     \begin{subfigure}[b]{0.15\textwidth}
         \centering
         \begin{ytableau}
            1 & 2 & 6 & 7 \\
            3 & 5 \\
            4
         \end{ytableau}
         \caption{}
         \label{fig:s-ytab}
     \end{subfigure}
     \hfill
     \begin{subfigure}[b]{0.15\textwidth}
         \centering
         \begin{ytableau}
            1 & 1 & 2 & 4 \\
            2 & 3 \\
            3
         \end{ytableau}
         \caption{}
         \label{fig:ss-ytab}
     \end{subfigure}
    \caption{Young diagrams and Young tableaux. Fig.~\ref{fig:y-diag} is \zwnn{the} Young diagram corresponding to the partition (4, 2, 1). Fig.~\ref{fig:s-ytab} is an associated standard Young tableau. Fig.~\ref{fig:ss-ytab} is an associated semistandard Young tableau with numbers 1 to 4.}
    \label{fig:ytab}
\end{figure}

Given a Young diagram $\lambda$, we can calculate $r_i$ and $l_i$, the number of standard Young tableaux and semistandard Young tableaux (with numbers from 1 to $d$), using the following formulae:
\begin{equation}
    r_\lambda = \frac{n!}{\prod_{(i,j)\in \lambda} h(i,j)}, \quad l_\lambda = \prod_{(i,j)\in \lambda}\frac{d+j-i}{h(i,j)}. \label{eq:hook}
\end{equation}
Here, $(i,j)$ refers to the box on the $i$-th row and $j$-th column, and $h(i,j)$ is the \emph{hook length} of the box $(i,j)$ in diagram $\lambda$, which is defined as the number of the boxes to the right of $(i,j)$, plus the number of boxes below $(i,j)$, plus 1. For example, $h(1,2) = 2+1+1=4$ for the Young diagram in Fig.~\ref{fig:y-diag}.

\zwnn{We need} asymptotic bounds for $l_\lambda$ and $r_\lambda$ when $n$ is large while $d$ remains constant. An upper bound for $l_\lambda$ could be found by lifting the constraints on columns for semistandard Young tableaux. On row $i$, the number of ways to fill in numbers is given by $\binom{\lambda_i+d-1}{d-1}$, which means that
\begin{equation}
    l_\lambda \le \prod_i \binom{\lambda_i+d-1}{d-1} = \poly(n).\label{eq:l}
\end{equation}
For $r_\lambda$ with $\lambda \vdash (n,d)$, note that
\begin{equation}
    \lambda_i - j + 1 \le h(i,j) \le \lambda_i - j + 1 + d
\end{equation}
because $\lambda_i - j$ is the number of boxes to the right of $(i,j)$, and the number of boxes below $(i,j)$ is between 0 and $d$. So
\begin{equation}
    r_\lambda \le \frac{n!}{\lambda_1! \lambda_2!\cdots \lambda_d!} = \binom{n}{\lambda_1,\cdots,\lambda_d} \label{eq:r-ub}
\end{equation}
and
\begin{equation}
    r_\lambda \ge \frac{n}{\frac{(\lambda_1+d)!}{d!}\frac{(\lambda_2+d)!}{d!}\cdots \frac{(\lambda_d+d)!}{d!}} = \binom{n}{\lambda_1,\cdots,\lambda_d}/\poly(n).\label{eq:r-lb}
\end{equation}

\subsubsection{Schur-Weyl duality}
\label{sec:duality}
Consider the Hilbert space $(\C^d)^{\otimes n}$, which corresponds to $n$ $d$-dimensional qudits. Let $\{|e_i\>\}_{i=1}^d$ be \zwmerge{an} orthonormal basis for each qudit. \zwmerge{A natural representation for the unitary group $SU(d)$ on $(\C^d)^{\otimes n}$ is given by the collective action of $U\in SU(d)$ on each subsystem,}
\begin{equation}
    (U^{\otimes n} ) |e_{i_1}\> \cdots |e_{i_n}\> = U|e_{i_1}\>  \otimes \cdots U|e_{i_n}\>.
\end{equation}
\zwmerge{On the other hand, the symmetric group $S_n$ naturally acts on $(\C^d)^{\otimes n}$ by permuting the subsystems.}  For $\pi \in S_n$, the corresponding operator $O_\pi$ behaves like
\begin{equation}
    O_\pi |e_{i_1}\> \cdots |e_{i_n}\> = |e_{i_{\pi^{-1}(1)}}\> \cdots |e_{i_{\pi^{-1}(n)}}\>.
\end{equation}
Intuitively, the operator $O_\pi$ moves qudit $i$ to the position $\pi(i)$ for each $i$.
Such representations for $SU(d)$ and $S_n$ are reducible. The Schur-Weyl duality states that the Hilbert space $(\C^d)^{\otimes n}$ could be decomposed in the following way:
\[
    (\C^d)^{\otimes n} = \bigoplus_{\lambda\vdash(n,d)} \cL_\lambda \otimes \cR_\lambda,
\]
so that $\cL_\lambda$ and $\cR_\lambda$ correspond to irreps of $SU(d)$ and $S_n$ respectively. In other words,
\begin{equation}
    U^{\otimes n} = \bigoplus_{\lambda\vdash(n,d)} U_\lambda \otimes I_{\cR_\lambda},\quad
    O_\pi = \bigoplus_{\lambda\vdash(n,d)}  I_{\cL_\lambda} \otimes \pi_\lambda,
\end{equation}
where $U_\lambda$ (resp.\ $\pi_\lambda$) is the operator corresponding to $U$ (resp.\ $\pi$) in the irrep labeled by $\lambda$. The dimensions for $\cL_\lambda$ and $\cR_\lambda$ are given by
\[
    \dim \cL_\lambda = l_\lambda,\quad \dim \cR_\lambda = r_\lambda 
\]
where $l_\lambda$ and $r_\lambda$ are defined in Eq.~\eqref{eq:hook}. we shall keep the dependence of $\cL_\lambda$ and $\cR_\lambda$ on $n$ implicit in the notation, as $n$ could be calculated as the total number of boxes in $\lambda$.

Now we define $\Pi_\lambda$ as the projector onto $\cL_\lambda \otimes \cR_\lambda$. In order to give an expression for $\Pi_\lambda$, we need to introduce the concept of \emph{Young symmetrizer}. Let $T$ be a standard Young tableaux that corresponds to Young diagram $\lambda$. Then $\cR(T)$ (resp.\ $\cC(T)$) is defined as the set of permutations that permutes numbers within each row (resp.\ column) of $T$. Now we define the Young symmetrizer $e_T$ as $e_T = r_T c_T$ where
\begin{equation}
    r_T = \sum_{\pi \in \cR(T)} O_\pi, \quad c_T = \sum_{\pi \in \cC(T)} \sgn(\pi) O_\pi,
\end{equation}
where $\sgn(\pi)$ is the sign of permutation $\pi$. Note that the Young symmetrizer is usually defined in the group algebra of $S_n$, but here we care only about the corresponding operator in $(\C^d)^{\otimes n}$. Then we have the formula for $\Pi_\lambda$
\begin{equation}
    \Pi_\lambda = \frac{r_\lambda^2}{n!^2}\sum_{\pi \in S_n} O_\pi e_T O_\pi^\dagger
\end{equation}
for any Young tableau $T$ that corresponds to Young diagram $\lambda$ \cite{goodman2009}.

Another interesting fact to \zwnn{note} is that for any standard Young tableau $T$ on Young diagram $\lambda$, $\frac{r_\lambda}{n!}e_T$ is a projector onto a subspace of $(\C^d)^{\otimes n}$ that corresponds to $U_\lambda$, the irrep of $SU(d)$ labeled by $\lambda$.

\subsubsection{Relating $n-1$ qudits to $n$ qudits}
\label{sec:lr}
Suppose that in $(\C^d)^{\otimes (n-1)}$, $P$ is a projector onto a subspace that corresponds to an irrep of $SU(d)$ labeled by $\lambda$. Then what do we know about $I_d \otimes P$, the operator in $(\C^d)^{\otimes n}$ obtained by adding a new qudit? Essentially we are asking about the decomposition of the tensor product of two irreps:
\begin{equation}
    U_{(1)} \otimes U_{\lambda} = \bigoplus_{\lambda'} U_{\lambda'}
\end{equation}
where $U_{(1)}$ is the $d$-dimensional fundamental representation of $SU(d)$. According to Littlewood--Richardson rule, each representation $U_{\lambda'}$ appears at most once on the right hand side, and the $\lambda'$, which appears could be obtained by adding a single box to $\lambda$. Such a relation is denoted by $\lambda \lhd \lambda'$, and we have
\begin{equation}
    U_{(1)} \otimes U_{\lambda} = \bigoplus_{\lambda':\lambda \lhd \lambda'} U_{\lambda'}.
\end{equation}

In other words, $I_d \otimes P$ is a sum of projectors into subspaces that correspond to irreps $U_{\lambda'}$, where $\lambda'$ satisfies $\lambda \lhd \lambda'$.

\subsection{$SU(d)$-covariant codes from \zwnew{$SU(d)$-symmetric} unitaries}
\label{sec:sud-construction}

Let $\zwnew{\cU_{\times,SU(d)}^n} \subset SU(d^n)$ be the set of unitary operators that commute with $U^{\otimes n}$ for all $U \in SU(d)$\zwmerge{, which forms a compact subgroup of $SU(d)$}. \zwnn{As discussed in}  Sec.~\ref{sec:duality}, any $V\in \cU_{\times,SU(d)}^n$ should take the form
\begin{equation}
    V = \bigoplus_{\lambda \vdash (n,d)} I_{\cL_\lambda} \otimes V_\lambda \label{eq:hx}
\end{equation}
for some unitaries $V_\lambda$. Let $H_{\times,SU(d)}^n$ be the Haar measure on $\cU_{\times,SU(d)}^n$. Then $V\sim H_{\times,SU(d)}^n$ will take the form in Eq.~\eqref{eq:hx} with each $V_\lambda$ drawn from the Haar measure of $SU(r_\lambda)$ independently.

\begin{defi}
A code is called $(n;\lambda)$-$SU(d)$ code for some $\lambda \vdash (n-1,d)$, if it encodes a qudit into $n$ qudits by first appending a $(n-1)$-qudit state $\rho_\lambda$ that lies within $\cL_\lambda \otimes \cR_\lambda$ and takes the form
\begin{equation}\label{eq:rho-lambda}
    \rho_\lambda = \frac{I_{\cL_\lambda}}{l_\lambda}\otimes |\psi_\lambda\>\<\psi_\lambda|
\end{equation}
for some pure state $|\psi_\lambda\> \in \cR_\lambda$, and then applying a unitary $V \in \cU_{\times,SU(d)}^n$. In particular, \zwnn{a \emph{$(n;\lambda)$-$SU(d)$ random code} is given by $V\sim H_{\times,SU(d)}^n$.} 
\end{defi}

Note that \zwnn{the choice of $|\psi_\lambda\>$ does not affect the definition of the $(n;\lambda)$-$SU(d)$ random code}, because different states in $\cR_\lambda$ are related by unitaries in $\cU_{\times,SU(d)}^n$, which could be absorbed in the Haar measure $H_{\times,SU(d)}^n$.

It is  straightforward to verify that the $(n;\lambda)$-$SU(d)$ codes satisfy the $SU(d)$-covariance conditions.

\begin{prop}
$(n;\lambda)$-$SU(d)$ codes are covariant with respect to $SU(d)$ symmetry, \zwnn{in the sense that for all $U\in SU(d)$,}
\[
    \cE(U\rho U^\dagger) = U^{\otimes n} \cE(\rho) (U^{\dagger})^{\otimes n}.
\]
where $\cE$ is the encoding channel of the code.
\end{prop}
\begin{proof}
First note that \zwnn{for $\rho_\lambda$ of the form Eq.~(\ref{eq:rho-lambda})}, it holds that
\begin{equation}
    [U^{\otimes (n-1)}, \rho_\lambda] = 0
\end{equation}
due to the Schur-Weyl duality. So we have
\begin{align}
    &U^{\otimes n} \cE(\rho) (U^{\dagger})^{\otimes n} \nonumber\\
    =& U^{\otimes n} V(\rho \otimes \rho_\lambda)V^\dagger (U^{\dagger})^{\otimes n}\nonumber\\
    =& V [U\rho U^\dagger \otimes U^{\otimes (n-1)}\rho_\lambda  (U^\dagger)^{\otimes (n-1)}]V^\dagger \nonumber\\
    =& V (U\rho U^\dagger \otimes \rho_\lambda )V^\dagger \nonumber\\
    =& \cE(U\rho U^\dagger).
\end{align}
\end{proof}

Again, the $(n;\lambda)$-$SU(d)$ random code \zwnn{can be regarded as a randomized} construction of $SU(d)$-covariant codes, \zwnn{and furthermore, it} indicates the  typical behaviors of all such codes.
It is also possible to define codes that encode $k>1$ qudits into $n$ qudits in a similar way, with logical operator $U^{\otimes k}$ mapped to $U^{\otimes n}$. Here we focus on the $k=1$ case due to its particular importance and leave $k>1$ for future studies.

\subsection{Performance of random $SU(d)$-covariant codes}
\label{sec:sud-error}
\zwnew{We now study the Choi error and worst-case error of $(n; \lambda)$-$SU(d)$ random codes against erasure.  Here we discuss the erasure of a fixed qudit, but note again that the results hold for the erasure of any qudit.} \zwnn{The complementary channel is then a partial trace over the other $n-1$ qudits, which we denote by $\Tr_{n-1}[\cdot]$.}
The erasure against $t>1$ qudits may be analyzed in a similar way, which we leave for future work.

\subsubsection{Choi error}

\begin{theo}
\label{thm:sud-choi}
In the large $n$ limit, if a partition $\lambda \vdash (n-1, d)$ satisfies $n-\lambda_1 = \Omega(n)$, then the expected Choi error of the $(n;\lambda)$-$SU(d)$ random code satisfies
\begin{equation}
    \E \epsilon_\choi \le \frac{\sqrt{d^2-1}}{2n} + O(n^{-2})
\end{equation}
against the erasure of a single qudit.
Furthermore, the probability that the Choi error of a $(n;\lambda)$-$SU(d)$ random code (with respect to $H^n_{\times,SU(d)}$) violates the inequality above is exponentially small, i.e.,
\begin{equation}
\Pr\left[\epsilon_\choi >  \frac{\sqrt{d^2-1}}{2n} + O(n^{-2})\right] = e^{-\Omega(n)}.
\end{equation}

\end{theo}
\begin{proof}
\zwnew{Similarly as the $U(1)$ case, the error can be analyzed using the complementary channel formalism. By Eq.~\eqref{eq:complementary},} for any fixed encoding unitary $V \in \cU_{\times,SU(d)}^n$ \zwnew{and $\rho_\lambda$ that satisfies Eq.~\eqref{eq:rho-lambda}}, we have
\begin{align}
    \epsilon_\choi =& \min_\zeta P\left(\Tr_{n-1}[V(|\hat\phi\>\<\hat\phi|\otimes \rho_\lambda)V^\dagger],\frac{I_d}{d}\otimes \zeta\right) \nonumber\\
    \le& P\left(\Tr_{n-1}[V(|\hat\phi\>\<\hat\phi|\otimes \rho_\lambda)V^\dagger],\Tr_{n-1}\Phi_\text{avg}\right) \nonumber\\
    &+\min_\zeta P\left(\Tr_{n-1}\Phi_\text{avg},\frac{I_d}{d}\otimes \zeta\right) \nonumber \\
    \le& \sqrt{2\left\|\Tr_{n-1}[V(|\hat\phi\>\<\hat\phi|\otimes \rho_\lambda)V^\dagger]-\Tr_{n-1}\Phi_\text{avg}\right\|_1}\nonumber \\
    &+\min_\zeta P\left(\Tr_{n-1}\Phi_\text{avg},\frac{I_d}{d}\otimes \zeta\right), \label{eq:sud-bound}
\end{align}
where \zwnn{$\rho_\lambda$ takes the form Eq.~(\ref{eq:rho-lambda}) and} \zwnew{$\Phi_\text{avg}$ is the \zwnnn{average physical state} given by} 
\begin{equation}
    \Phi_\text{avg} = \E_{V\sim H^n_{\times,SU(d)}} V(|\hat\phi\>\<\hat\phi|\otimes \rho_\lambda)V^\dagger.
\end{equation}
After averaging over $V$ sampled from $H^n_{\times,SU(d)}$, the first term in Eq.~\eqref{eq:sud-bound} can be bounded using the partial decoupling theorem,
\begin{align}
    &\E_{V\sim H^n_{\times,SU(d)}} \sqrt{2\left\|\Tr_{n-1}[V(|\hat\phi\>\<\hat\phi|\otimes \rho_{\lambda})V^\dagger]-\Tr_{n-1}\Phi_\text{avg}\right\|_1} \nonumber \\
    \le& \sqrt{2\E_{V\sim H^n_{\times,SU(d)}}\left[\left\|\Tr_{n-1}[V(|\hat\phi\>\<\hat\phi|\otimes \rho_{\lambda})V^\dagger]-\Tr_{n-1}\Phi_\text{avg}\right\|_1\right]} \nonumber \\
    \le & \sqrt{2} \times 2^{-\frac{1}{4}H_{\min}(A^*|RE)_{\Lambda(\Psi,\cT)}},
\end{align}
\zwnn{where $\Psi = |\hat\phi\>\<\hat\phi|\otimes \rho_\lambda$  is the initial state  on $AR$, and $\cT^{A\rightarrow E}$ is the complementary erasure $\Tr_{n-1}$  (detailed definitions introduced in Thm.~\ref{thm:choi} and Sec.~\ref{sec:partial-dec}).}
{We prove in Appendix~\ref{app:sud-entropy} that} $H_{\min}(A^*|RE)_\Lambda=\Omega(n)$, so the first term in Eq.~\eqref{eq:sud-bound} is exponentially small in $n$.
\lhn{
The main observation here is that the support of $\Psi^A \equiv \Tr_R[\Psi^{AR}]$ lies within the subspaces corresponding to $\lambda'$ where $\lambda \lhd \lambda'$, according to Section~\ref{sec:lr}. As a result, one could bound the norm of $\Lambda(\Psi,\cT) = \Xi^{A\bar A \to A^*}(\Psi^{AR} \otimes \tau^{\bar AE})(\Xi^\dagger)^{A\bar A \to A^*}$. Then we could use the bounds on $l_\lambda$ and $r_\lambda$ in Section~\ref{sec:duality} to get a bound on the conditional min-entropy.
}

Now we take $\zeta = I_d/d$. As shown in Appendix~\ref{app:avg},
\begin{equation}
    \Tr_{n-1}\Phi_{\avg}  = \frac{1}{n}|\hat\phi\>\<\hat\phi| + \frac{n-1}{n} \frac{I_d}{d} \otimes \frac{I_d}{d},
\end{equation}
and by a straightforward calculation we have
\begin{equation}
    P\left(\Tr_{n-1}\Phi_\text{avg},\frac{I_d}{d}\otimes \frac{I_d}{d}\right) = \frac{\sqrt{d^2-1}}{2n}+O(n^{-2}).
\end{equation}

\end{proof}


\subsubsection{Worst-case error}

\begin{theo}
\label{thm:sud-worst}
In the large $n$ limit, if a partition $\lambda \vdash (n-1, d)$ satisfies $n-\lambda_1 = \Omega(n)$, then the expected worst-case error of the $(n;\lambda)$-$SU(d)$ random code satisfies
\begin{equation}
    \E \epsilon_\worst \le \frac{\sqrt{d^2-1}}{2n} + O(n^{-2})
\end{equation}
against the erasure of a single qudit.
Furthermore, the probability that the worst-case error of a $(n;\lambda)$-$SU(d)$ random code (with respect to $H^n_{\times,SU(d)}$)  violates the  inequality above is exponentially small, i.e.,
\begin{equation}
\Pr\left[\epsilon_\worst >  \frac{\sqrt{d^2-1}}{2n} + O(n^{-2})\right] = e^{-\Omega(n)}.
\end{equation}

\end{theo}
\begin{proof}

For the worst-case error, first note that the reference system $R$ could always be assumed to be $d$-dimensional, as mentioned in Sec.~\ref{sec:approx-qec}. Then we have
\begin{align}
    \epsilon_\worst =& \min_\zeta \max_{|\psi\>} P\left(\Tr_{n-1}[V(|\psi\>\<\psi|\otimes \rho_\lambda)V^\dagger],\frac{I_d}{d}\otimes \zeta\right) \nonumber\\
    \le& \max_{|\psi\>}P\left(\Tr_{n-1}[V(|\psi\>\<\psi|\otimes \rho_\lambda)V^\dagger],\Tr_{n-1}\Phi_\text{avg}\right)\nonumber \\
    &+\min_\zeta \max_{|\psi\>} P\left(\Tr_{n-1}\Phi_\text{avg},\frac{I_d}{d}\otimes \zeta\right) \nonumber \\
    \le& \max_{|\psi\>}\sqrt{2\left\|\Tr_{n-1}[V(|\psi\>\<\psi|\otimes \rho_\lambda)V^\dagger]-\Tr_{n-1}\Phi_\text{avg}\right\|_1}\nonumber \\
    &+\min_\zeta \max_{|\psi\>} P\left(\Tr_{n-1}\Phi_\text{avg},\frac{I_d}{d}\otimes \zeta\right),
\end{align}
Again the first term is exponentially small in $n$ as shown in Appendix~\ref{app:sud-entropy}. If we take $\zeta=I_d/d$, according to Appendix~\ref{app:worst-case} we can see that the purified distance is maximized by $|\psi\> = |\hat\phi\>$, which implies that
\begin{equation}
    \E \epsilon_\worst \le \frac{\sqrt{d^2-1}}{2n} + O(n^{-2}).
\end{equation}
\end{proof}

\zwnn{Note that {here we did not use Lemma~\ref{lem:worst-bound} as in Sec.~\ref{sec:worst-case} for the $U(1)$ case, because} for $(n; \lambda)$-$SU(d)$ random codes, $\rho^{x,x'}_\avg$ (the average of $\rho^{x,x'}$ defined in Eq.~\eqref{eq:rho-def} over the random unitary) is nonzero for $x\not= x'$, which could only lead to a weaker bound of $\E \epsilon_\worst=O(n^{-1/2})$.}

\subsection{Comparisons with fundamental limits and known protocols}
\label{sec:sud-compare}

For $SU(d)$-covariant codes, the bound given in Ref.~\cite{faist2019} is
\begin{align}
    \epsilon_\choi &\ge \frac{1}{dn},\\ \epsilon_\worst &\ge \frac{1}{2n}.
\end{align}
This bound could be obtained either from their Thm.~1 or Appendix~E (note that their Thm.~6 only gives a weaker bound). From Sec.~\ref{sec:sud-error}, with high probability $(n; \lambda)$-$SU(d)$ random codes satisfy
\begin{align}
    \epsilon_\choi &\le \frac{\sqrt{d^2-1}}{2n},\\ \epsilon_\worst &\le \frac{\sqrt{d^2-1}}{2n}
\end{align}
as long as $n-\lambda_1 = \Omega(n)$. We can see that our code saturates this bound up to a constant factor.

It is interesting to compare our construction with the ``generalized $W$-state encoding'' given in Ref.~\cite[Sec.~VII.B]{faist2019}, where the encoding channel is given by
\begin{equation}
    |{\psi}\> \to \frac{1}{\sqrt{n}}(|{\psi,\perp,\cdots,\perp}\> + |{\perp,\psi,\perp,\cdots}\> + \cdots + |{\perp,\cdots,\perp,\psi}\> ),
\end{equation}
\zwnn{mapping} a $d$-dimensional logical qudit into $n$ $(d+1)$-dimensional qudits. The output of the complementary channel of \zwnn{this} encoding and erasure is quite similar to the \zwnnn{average physical state} \zwnn{of} our encoding derived in Appendix~\ref{app:avg}, with the maximally mixed state replaced by $|{\perp}\>\<{\perp}|$. As a result, our proof in Sec.~\ref{sec:sud-error} shows that the Choi and worst-case errors of this encoding is also $\frac{\sqrt{d^2-1}}{2n}$. Our $(n;\lambda)$-$SU(d)$ codes have higher efficiency than the generalized $W$-state encoding in that we need $n$ $d$-dimensional qudits instead of $(d+1)$-dimensional qudits. Also, this homogeneity property \lhn{(that the logical qudit and the physical qudit have the same dimension)} might also be helpful for constructing fault tolerance schemes.

\lh{Other constructions of $SU(d)$-covariant codes could be found in Refs.~\cite{Woods2020continuousgroupsof,YangWoods,wang2020,wang2021theory}.  In Ref.~\cite{Woods2020continuousgroupsof} a construction based on quantum reference frames was given, but the worst-case error rate was $O(\operatorname{polylog}(n)/n)$, which is asymptotically larger than our code.}  \lh{Another  construction was provided in Ref.~\cite{YangWoods}, which also has $O(1/n)$ worst-case error \zwnn{but with larger factors }\lhn{and worse scaling in $d$}. \zwnnn{The edge valence-bond-solid code \cite{wang2020,wang2021theory} provides another example of $SU(d)$-covariant codes where a logical qudit is encoded into a $d$-dimensional qudit together with $n$ $(d^2-1)$-dimensional qudits, for which the error scaling is again worse than our construction.}  
}










\section{On unitary designs and random circuits with symmetries}
\label{sec:designs}
\zw{In the above, we considered  Haar-random unitaries \zwnn{with $U(1)$ and $SU(d)$ symmetries}, for which one important motivation is to understand the typical performance of all such unitaries.     A key follow-up question of both practical and mathematical interest is how well the  results  hold for more ``efficient'' versions of random unitaries such as unitary $t$-designs (``pseudorandom'' distributions that match the Haar measure up to $t$ moments) and random local quantum circuits (circuits composed of random local gates). 
\zwnn{In particular, the random circuit models are broadly important in that they provide a powerful lens into the complexity and dynamics (especially early-time physics) of physical systems by capturing the locality of interactions.}
Consider the case without symmetries---it is known that the decoupling and QEC properties of Haar-random unitaries hold for (approximate) 2-designs \cite{dupuis2014}, and that random circuits converge to $t$-designs in depth polynomial in $t$ and $n$ \cite{harrow2009,brandao2016}, which imply that random circuits can provide rather efficient constructions of good codes (see also Refs.~\cite{BrownFawzi:decoupling,Brown13shortrandom,Gullans2021}).    Do analogous conclusions hold for the case with symmetries?

}
In our analysis, the Haar randomness has only been used in the partial decoupling theorem, and as was noted in Ref.~\cite{wakakuwa2019}, 
\zw{symmetric 2-designs 
are sufficient for  the partial decoupling bounds to hold. Therefore, all our results hold for symmetric 2-designs. 
}
Although 2-designs 
for the full unitary group  have been widely studied \cite{harrow2009,brandao2016,cleve2015,nakata2017}, little is known about 2-designs (let alone higher-order designs) with symmetries, as is in our consideration. 
Particularly for the fundamental problem of convergence of random circuits to designs with $U(1)$ symmetry \zwnn{or charge conservation}, we \zwnn{explicitly} point out a few interesting  differences.
\zwnn{For the no-symmetry case, repeated applications of local  random unitary gates converge to designs \cite{harrow2009,brandao2016}, but it appears difficult to adapt the techniques in Refs.~\cite{harrow2009,brandao2016} to establish an analogous result even for 2-designs in the case with symmetries}---negative values could appear in the operator basis, making the Markov chain analysis difficult; The conversion into the Hamiltonian spectrum problem discussed in Ref.~\cite{brandao2016} does not work here either, due to the complicated structure of the eigenspaces of Hamming weight operator. \zw{To summarize, there seem to be nontrivial obstacles to adapting previous proofs of convergence of random circuits to 
the case with symmetries, \zwnn{which are worth further understanding}, and it remains open how to efficiently construct symmetric 2-designs.  }  
\zw{\zwnn{In fact}, it is recently shown that \cite{marvian2020locality}  in the presence of continuous symmetries, the group of unitaries generated by local symmetric gates is in general a proper subgroup of the group of  global symmetric ones.  As a result, local random circuits cannot converge to the Haar measure with symmetries, and it remains to be further studied whether and how fast they converge to certain $t$-designs. 
Note that for $SU(d)$, it is known that \cite{hulse2021}  2-local symmetric unitaries could generate unitaries in $\cU_{\times,SU(d)}^n$ up to relative phases when $d=2$, while the generated group do not form a 2-design when $d>2$. 
Also note that for $SU(d)$, the quantum Schur transform \cite{bacon2006} might be helpful for constructing 2-designs for $H_{\times,SU(d)}^n$, by converting the system into the basis that corresponds to Schur-Weyl duality. It might be easier to construct 2-designs in this  basis, but we leave \zwnn{a more careful study for} future work.
For the weaker QEC property \zwnn{(which is implied by convergence results)}, given the dominance of the symmetry terms at late times,  we conjecture that
random circuits composed of symmetric local gates are able to approach the near-optimal performance of Haar-random symmetric unitaries derived in Thm.~\ref{thm:choi}, Thm.~\ref{thm:worst}, Thm.~\ref{thm:sud-choi} and Thm.~\ref{thm:sud-worst} with an efficiency no worse than the no-symmetry case, \zwnnn{e.g.,~$\widetilde{O}(n)$ gates or $O(\mathrm{polylog}(n)$) depth for circuit architecture without geometries \cite{BrownFawzi:decoupling}. }
}

\section{Discussion and outlook}
\label{sec:discussion}

In this work, we rigorously studied  $U(1)$- and $SU(d)$-covariant codes generated by Haar-random unitaries with corresponding symmetries, \zwnnn{which have simple structures and faithfully represent typical features of symmetric unitaries}.  A central message is that, with overwhelming probability, the \zwnew{error rates of such codes under erasure noise as measured by both Choi and worst-case purified distances can scale as $O(n^{-1})$ in the number of physical subsystems $n$ which \zwnew{nearly} saturate known lower bounds to leading order, \zwnnn{indicating the near-optimality of both  the lower bounds and our randomized code constructions. }} 

\zwnew{
How does our \zwnn{analysis} apply to the case without symmetries, \zwnn{where} our code constructions are modified by replacing the  Haar-random \zwnn{symmetric} unitary by a fully Haar-random one? As long as the quantum \zwnn{Singleton} bound $n-k \ge 4t$ is satisfied, the resulting random code has an error rate exponentially small in $n$, in contrast to polynomial small for the case with symmetries. 
To be more explicit, let $\Delta = n-k-4t$, then the random code will have expected Choi error $e^{-\Omega(\Delta)}$. 
 This could be shown using an analysis analogous to the case with symmetries, and the main difference is that the symmetry terms (e.g., the second term in Eq.~\eqref{eq:bound}) are naturally zero in this case and the error solely comes from the decoupling bound \cite{dupuis2014}. 
This comparison also gives an intuition for the Eastin--Knill theorem and the lower bounds for covariant codes from a mathematical perspective.}

\zwnew{We would also like to comment on the noise model. 
Note that although our analyses are presented in terms of the erasure of \zwnn{a specific set of} $t$ qudits, the results hold for the more general case of the erasure of any combination of $t$ qudits with some probability. To see this, first note that the purified distance is independent of the choice of $t$ qudits, because the permutation of the qudits \lhn{\zwnnn{belongs to the groups of covariant unitaries, namely} $\cU_{\times,U(1)}$ for $U(1)$ and $\cU_{\times,SU(d)}^n$ for $SU(d)$, }  and thus could be absorbed into the Haar measure \zwnnn{over the groups}. Then a union bound could be used over all choices of $t$ qudits, which {at most amplifies} the failure probability by a factor of $\binom{n}{t}$. Since $t=o(n)$ is  needed for \zwnn{meaningful results}, the failure probability is still exponentially small as $\log\binom{n}{t}=o(n)$.
Instead of erasing $t$ out of the $n$ qudits, a natural and stronger model of erasure noise is to have each qudit erased with some independent probability. Our analysis can be easily applied to this model by combining our bounds with the distribution on the number of qudits erased.  
}

\zwnnn{
In the paper, we already mentioned a few specific problems for future work, including the cases of multiple logical qudits and multiple erasure noise for $SU(d)$. Our constructions and results may also have applications for fault tolerance given the imposed transversality feature, which could be interesting to explore.  Note also that the analysis here may shed new light on the trade-off between  symmetry properties and QEC accuracy studied recently in Refs.~\cite{liu2021quantum,liu2021approximate} via the behavior of the error terms when symmetry constraints are relaxed. 

\smallskip

\textbf{Remarks on relevance to physics.} We expect our settings, results, and methods to find broad applications in physics through several directions with intimate connections to QEC. 
Note also that QEC properties of quantum systems go hand in hand with their entanglement properties.
Here we discuss preliminarily the potential relevance and point out some references, leaving in-depth explorations for future work.  Remarkably, random unitaries and circuits have drawn great interest in recent years as solvable models of complex, chaotic quantum systems which are key to quantum gravity and many-body physics.  Given the fundamental importance of symmetries and conservation laws, the symmetric versions of random unitaries and circuits that we consider here are expected to be broadly relevant in physical contexts. More specifically, our study of their QEC properties may find implications to many-body physics and quantum gravity through these lenses: 
\begin{itemize}
    \item 
    \zwnnn{Charged black holes and Hayden--Preskill thought experiment with symmetries. }  Hayden and Preskill  \cite{hayden2007} considered the retrieval of quantum information from black hole radiation based on the scrambling feature of quantum black holes modeled by random circuits, which provides important insights to and has stimulated many recent developments on the black hole information problem and quantum gravity. 
    In the original model, the information recovery essentially relies on Haar-random codes, which have almost optimal QEC properties.  However, conserved charges are expected to place certain obstructions on the recovery, for which our results of  random symmetric codes  may indicate various quantitative characterizations  (see also recent works Refs.~\cite{Yoshida:softmode,liu2020,nakata2021black,tajima2021symmetry} that studied Hayden--Preskill with charge conservation from different aspects).
    Note that our current analysis and the optimality arguments mostly focus on the regime of relatively small $t$ (size of erasure), so in order to understand the connections to Hayden--Preskill it could be important to further look into the large $t$ regime. 
    Furthermore, note that the decay process and the Hayden--Preskill recovery of charged black holes is also closely related to other key clues about quantum gravity such as weak gravity \cite{Arkani_Hamed_2007} and no-global-symmetry \cite{BanksSeiberg11} conjectures, thus our analysis may also lead to useful quantitative statements in these regards.
  
\item Scrambling, entanglement growth, and \zwnnn{emergent QEC in complex quantum many-body systems.} Random circuit models have also drawn great interest in condensed matter physics, 
leading to highly active directions like entanglement and operator growth \cite{Nahum2,Nahum1}, and measurement-induced entanglement transition \cite{PhysRevX.9.031009,PhysRevB.100.134306,PhysRevB.99.224307}.  Conservation laws lead to diffusive transport of the conserved quantities so that the laws of information scrambling (which underpins QEC properties) including operator spreading \cite{PhysRevX.8.031057,PhysRevX.8.031058} and R\'enyi entanglement entropy growth \cite{Rakovszky19:renyi,Znidaric2020,Huang20} could exhibit features fundamentally different  from the case without symmetries.  Specifically, note that the deviation from maximal entanglement in quantum many-body systems with conservation laws discussed in e.g.~Ref.~\cite{Huang22} may be closely related to the QEC error of covariant codes, since both have origins in the logical charge information contained in subsystems (e.g., reflected by the symmetry terms). Moreover, the so-called monitored random circuits in which random unitary gates are interspersed with measurements exhibit phase transitions when tuning the measurement density that can be understood from QEC properties \cite{ChoiBaoQiAltman,GullansHuse20:purification}, indicating interesting connections between quantum codes and phases of matter.  This work essentially addresses the late-time (equilibrium) and low measurement density limits of symmetric random circuit models, and it would be interesting to further understand the complete behaviors in early-time regimes and for different measurement densities, in particular, the time scales and measurement densities for achieving the optimal error scaling in certain circuit architectures.  


\end{itemize}

In addition, as mentioned, symmetries and QEC underlie a series of key recent developments in holography and AdS/CFT as well.  In particular, the famous conjecture about quantum gravity that exact global symmetries are not allowed is recently justified in AdS/CFT based on the QEC formulation of AdS/CFT \cite{HarlowOoguri2018arXiv181005338H,harlow2019}, and the argument at least for continuous symmetries indeed has deep connections to the limitations of covariant codes \cite{faist2019}.  Note that the transversality feature deduced from entanglement wedge reconstruction arguments \cite{czech12,wall2014,headrick2014,jafferis2016,PhysRevLett.117.021601,PhysRevX.9.031011} is a critical component of the argument, thus our study (in particular, $SU(d)$) may 
motivate natural models that exhibit optimal QEC (reconstruction) properties, and lead to quantitative insights into the corrections to the exact QEC  or transversality conditions needed to ensure consistency (e.g.,~via the random tensor network models \cite{Hayden2016}).

}


\zw{
With this work, we hope to stimulate further study into random unitaries and circuits with symmetries, which, as discussed above, exhibit many intriguing distinctions from the no-symmetry case.  In this work, we considered random global unitaries, and it would be important to further study random circuits since they can capture the ``complexity'' of the construction as well as the locality structure that is important in physical \zwnnn{and practical} scenarios.     As a general program, it would be interesting to better understand the relations between Haar-random unitaries, $t$-designs, and random local circuits (with different architectures), in the presence of various kinds of symmetries, through various kinds of physical properties and measures.  Here we take a first step by analyzing the QEC performance of random unitaries with $U(1)$ and $SU(d)$ symmetries, and conjectured that it holds for low-depth \zwnnn{symmetric} random circuits.   As discussed,  \zwnnn{there appear to be fundamental difficulties in adapting known methods to achieve a full understanding of the convergence of symmetric random circuits to designs}, but it would already be interesting and useful to look into the behaviors of ``measures'' of scrambling and randomness, such as frame potentials \cite{Scott_2008,2016arXiv160908172Z,RobertsYoshida}, out-of-time-order correlators \cite{hosur2016,RobertsYoshida}, R\'enyi entanglement entropies \cite{hosur2016,LLZZ18,PhysRevLett.120.130502}, which are widely used in physics and quantum information.
}


\section*{Acknowledgements}
We thank Daniel Gottesman, Aram Harrow, Zhi Li, Sirui Lu, Shengqi Sang, Jon Tyson, Beni Yoshida, Sisi Zhou for useful discussions and feedback.   
LK is supported by NSF grants No. CCF-1452616 and No. OMA-2016245. ZWL is supported by Perimeter Institute for Theoretical Physics.
Research at Perimeter Institute is supported in part by the Government of Canada through the Department of Innovation, Science and Economic Development Canada and by the Province of Ontario through the Ministry of Colleges and Universities.

\appendix


\section{Conditional min-entropy}

In this appendix we \zwnn{present in detail our derivation of the conditional min-entropy bounds.} We shall first present a general lower bound, and then \zwnn{specifically} apply it to the $U(1)$ and $SU(d)$ cases. 

\begin{lemma}
\label{lem:entropy}
For any bipartite positive operator $\rho^{PQ}$, we have\lhn{
\begin{equation}
     -\log \|\rho^{PQ}\|_\infty - \log\dim Q \le H_{\min}(P|Q)_\rho \le -\log \|\rho^{PQ}\|_\infty.
\end{equation}}
\end{lemma}
\begin{proof}
\zwnnn{Denote $s\equiv\|\rho^{PQ}\|_\infty$. } We can take $\sigma^Q=I^Q/\dim Q$ and $\lambda = -\log s - \log\dim Q$. Then \zwnn{it holds} that 
\begin{equation}
    2^{-\lambda} I^P \otimes \sigma^Q = s I^{PQ} \ge \rho^{PQ}. \label{eq:hmin-condition}
\end{equation}
\lhn{
By definition, $H_{\min}(P|Q)_\rho $ is given by the supremum of $\lambda$ such that $2^{-\lambda} I^P \otimes \sigma^Q \ge \rho^{PQ}$ for some positive operator $\sigma^Q$ with trace equal to 1. Since our specific choice of $\sigma$ and $\lambda$ satisfies this condition, we have
\begin{equation}
    H_{\min}(P|Q)_\rho \ge -\log \|\rho^{PQ}\|_\infty - \log\dim Q.
\end{equation}
}

For the other \zwnn{inequality}, note that the maximum eigenvalue \zwnnn{of} $2^{-\lambda}I^P \otimes \sigma^Q$ is at most $2^{-\lambda}$, so in order to have $2^{-\lambda}I^P \otimes \sigma^Q \ge \rho^{PQ}$, we must have $2^{-\lambda} \ge s$.
\end{proof}

Then we have the following theorem that gives a general lower bound for $H_{\min}(A^*|RE)_{\Lambda(\Psi,\tau)}$ for the general direct-sum-product decomposition
\begin{equation}
    \H^A = \bigoplus_j \H_j^{A_l} \otimes \H_j^{A_r}. \label{eq:dsp}
\end{equation}
Recall that $\Lambda(\Psi,\tau) = \Xi^{A\bar A \to A^*}(\Psi^{AR} \otimes \tau^{\bar AE})(\Xi^\dagger)^{A\bar A \to A^*}$ where
\begin{equation}
    \Xi^{A \bar{A} \rightarrow A^{*}}:=\bigoplus_{j} \sqrt{\frac{d_A l_j}{r_j}}\<\Phi_j^l|^{A_l\bar A_l}\left(\Pi_{j}^{A} \otimes \Pi_{j}^{\bar{A}}\right).
\end{equation}

\begin{theo}
\label{lem:entropy-lb}
Suppose that the support of $\Psi^{A}:= \Tr_R\Psi^{AR}$ lies within $\bigoplus_{j\in \cJ} \H_j^{A_l} \otimes \H_j^{A_r}$ for some set $\cJ$. If the reference system is $k$ qudits and the channel $\tau$ is erasure of $n-t$ qudits, then
\begin{equation}
    H_{\min}(A^*|RE)_{\Lambda(\Psi,\tau)} \ge -(2t+k)\log d - \log \max_{j \in \cJ}\frac{ l_j}{r_j}
\end{equation}
\end{theo}
\begin{proof}
Note that $\Pi_j^A \Psi^{AR} = \Psi^{AR}\Pi_j^A = 0$ for any $j \not\in \cJ$, we can replace $\Xi^{A \bar{A} \rightarrow A^{*}}$ by $\tilde \Xi^{A \bar{A} \rightarrow A^{*}}$ defined as
\begin{equation}
    \tilde \Xi^{A \bar{A} \rightarrow A^{*}}:=\bigoplus_{j\in\cJ} \sqrt{\frac{d_A l_j}{r_j}}\<\Phi_j^l|^{A_l\bar A_l}\left(\Pi_{j}^{A} \otimes \Pi_{j}^{\bar{A}}\right),
\end{equation}
and
\begin{equation}
    \|\tilde \Xi\|_\infty \le \max_{j \in \cJ}\sqrt{\frac{d_A l_j}{r_j}} = \max_{j \in \cJ}\sqrt{\frac{d^n l_j}{r_j}}.
\end{equation}

The state $\tau^{\bar AE}$ is the Choi state that corresponds to the erasure of $n-t$ qudits, which \lhn{is} equal to the maximally mixed state on $n-t$ qudits together with $t$ EPR pairs 
\lhn{
\begin{equation}
    \tau^{\bar AE} = \left(\frac{I}{d}\right)^{\otimes (n-t)} \otimes (|\hat\phi\>\<\hat\phi|)^{\otimes t}
\end{equation}
where $|\hat\phi\>$ is the maximally entangled state.
}

We can see that
\begin{equation}
    \|\Psi^{AR} \otimes \tau^{\bar AE}\|_\infty \le \frac{1}{d^{n-t}},
\end{equation}
so
\begin{equation}
    \|\Lambda(\Psi,\tau)\|_\infty \le \|\tilde \Xi\|_\infty^2 \|\Psi^{AR} \otimes \tau^{\bar AE}\|_\infty = d^{t} \max_{j \in cJ}\frac{ l_j}{r_j}.
\end{equation}
Then by Lemma~\ref{lem:entropy} we have
\begin{equation}
    H_{\min}(A^*|RE)_{\Lambda(\Psi,\tau)} \ge -(2t+k)\log d - \log \max_{j \in \cJ}\frac{ l_j}{r_j}.
\end{equation}

\end{proof}

\subsection{$U(1)$ case}
\label{app:entropy}

For the $U(1)$ case, $0\le j \le n$ labels the eigenspaces of Hamming weight operator and $l_j = 1, r_j = \binom{n}{j}$. For $(n, k; \alpha)$-$U(1)$ codes, we can set $\cJ = \{j|\alpha \le j \le k + \alpha\}$. When studying Choi error, the reference system have $k$ qubits (note that $d=2$), so in the large $n$ limit,
\begin{align}
    &H_{\min}(A^*|RE)_{\Lambda(\Psi,\tau)} \nonumber\\
    \ge& -(2t+k) - \log \max_{j \in \cJ}\frac{1}{\binom{n}{j}} \nonumber\\
    =& n\min\left\{H_b\left(\frac{\alpha}{n}\right),H_b\left(\frac{\alpha+k}{n}\right)\right\} - (2t+k) - O(\log n).
\end{align}
where $H_b(p) = -p\log p - (1-p)\log (1-p)$ is the binary Shannon entropy.

When studying worst-case error using Lemma~\ref{lem:worst-bound}, the reference system $R$ is trivial, and
\begin{align}
    H_{\min}(A^*|RE)_{\Lambda(\Psi,\tau)} \ge & n\min\left\{H_b\left(\frac{\alpha}{n}\right),H_b\left(\frac{\alpha+k}{n}\right)\right\} \nonumber\\
    &- 2t - O(\log n).
\end{align}
\subsection{$SU(d)$ case}
\label{app:sud-entropy}

For the $SU(d)$ case, the labels $j$ are the partitions of $n$ into at most $d$ parts. When studying $(n; \lambda)$-$SU(d)$ codes, according to Sec.~\ref{sec:lr}, the support of $\Psi^A$ only overlaps with $\Pi_{j}$ that satisfies $\lambda \lhd j$, i.e., $\cJ = \{j|\lambda \lhd j\}$.

Consider the large $n$ limit with $d$ remaining constant. Suppose that $\lambda$ satisfies $n-\lambda_1 = \Theta(n)$, from Eq.~\eqref{eq:l} and Eq.~\eqref{eq:r-lb} we know that
\begin{equation}
     H_{\min}(A^*|RE) \ge -2t\log d + n H(\lambda) - O(\log n)
\end{equation}
where
\begin{equation}
    H(\lambda) = H\left(\frac{\lambda_1}{n-1},\cdots, \frac{\lambda_d}{n-1}\right)
\end{equation}
is the Shannon entropy of the distribution $(\frac{\lambda_1}{n-1},\cdots, \frac{\lambda_d}{n-1})$. Note that the difference resulted from adding a box to $\lambda$ is small and could be absorbed into $O(\log n)$.

\section{Alternative derivation of $U(1)$ conditional min-entropies}
\label{sec:alternative}
In this appendix we provide an alternative \zwnn{analysis of the conditional min-entropies} in the $U(1)$ case, \zwnn{which could be of independent interest}. \zwnn{In fact we give  exact expressions of} $\Lambda(\Psi,\tau)$,  \zwnn{which yield lower bounds} as well as \zwnn{upper bounds} for \zwnn{$H_{\min} $}. 

\subsection{Choi}\label{app:alternative_choi}
We first restate the partial decoupling theorem adapted to our structure of the space, which corresponds to $l_j=1$ and $r_j = \binom{n}{j}$ for all $0\le j \le n$ in Ref.~\cite{wakakuwa2019}. We denote by $\Pi_j$ the projector onto the subspace with Hamming weight $j$.
\begin{lemma}[Partial Decoupling]
Let $\mathcal{T}^{A\to E}$ be any channel mapping system $A$ to system $E$, and let $\Psi^{AR}$ be any joint state of system $A$ and $R$. We have
\begin{align}
    &\mathbb{E}_{U \sim H_{\times,U(1)}}\left[\left\|\mathcal{T}^{A \rightarrow E} \circ \mathcal{U}^{A}\left(\Psi^{A R}\right)-\mathcal{T}^{A \rightarrow E}\left(\Psi_{\avg}^{A R}\right)\right\|_{1}\right] \nonumber\\
    \leq& 2^{-\frac{1}{2} H_{\min{} }(A^*|RE)_{\Lambda(\Psi, \mathcal{T})}},
\end{align}
where
\begin{equation}
    \Psi_\avg^{A R} = \E_{U\sim H_{\times,U(1)}} U \Psi^{A R} U^\dagger.
\end{equation}
The state $\Lambda(\Psi, \mathcal{T})$ is defined as
\begin{equation}
    \Lambda(\Psi, \mathcal{T}) = \Xi(\Psi^{AR} \otimes \tau^{\bar A E})\Xi^\dagger
\end{equation}
where $\tau^{\bar A E}$ is the Choi state of $\mathcal{T}$ and the operator $\Xi^{A\bar A \to A^*}$ is
\begin{equation}
    \Xi^{A \bar{A} \rightarrow A^{*}}:=\bigoplus_{j=0}^n \sqrt{\frac{2^n }{\binom{n}{j}}}\left(\Pi_{j}^{A} \otimes \Pi_{j}^{\bar{A}}\right).
\end{equation}
\end{lemma}

For simplicity\lhn{, for any integer $m$} we define the $2m$-qubit normalized state
\begin{equation}
    |\phi_i^{(m)}\> = \binom{m}{i}^{-1/2} \sum_{v \in \{0,1\}^m,\,|v|=i}|v\>|v\>,
\end{equation}
which is the maximally entangled state between two copies of subspaces with Hamming weight $i$ of $m$ qubits. We also define $|\hat\phi^{(m)}\>$ as $m$ EPR pairs.

In our setting the state $\Psi^{AR}$ is the encoded state before applying the random unitary, which is $k$ EPR pairs appended by the fixed state $|\psi\>$,
\begin{equation}
    \Psi^{AR} = |\Psi^{AR}\>\<\Psi^{AR}|,\quad |\Psi^{AR}\> =  |\hat \phi\>^{A_1 R} \otimes |\psi_\alpha\>^{A_2}.
\end{equation}
$A_1$ and $A_2$ refers to different parts of $A$, and have $k$ and $n-k$ qubits each. Then it is easy to see that
\begin{equation}
    \Pi_j^A |\Psi\>^{AR} = \sqrt{\frac{\binom{k}{j-\alpha}}{2^k}}|\phi_{j-\alpha}^{(k)}\>^{A_1R}|\psi_\alpha\>^{A_2}.
\end{equation}

The channel $\mathcal T$ traces over $n-t$ qubits, so the corresponding \zwnn{Choi state} is
\begin{equation}
    \tau^{\bar A E} = \Tr_{n-t} |\hat\phi^{(n)}\>\<\hat\phi^{(n)}| = |\hat\phi^{(t)}\>\<\hat\phi^{(t)}|^{\bar A_1 E} \otimes \frac{I^{\bar A_2}}{2^{n-t}}.
\end{equation}
Let $\Pi^{(a)}_b$ be the subspace on $a$ qubits with Hamming weight $b$ \lhn{for integers $a$ and $b$}. We have
\begin{align}
    &\Pi_j^{\bar A}\tau^{\bar A E}\Pi_{j'}^{\bar A} \nonumber \\
    = & \frac{1}{2^{n-t}}\sum_i\Pi_j^{\bar A}\left[|\hat\phi^{(t)}\>\<\hat\phi^{(t)}|^{\bar A_1 E} (\Pi^{(n-t)}_i)^{\bar A_2}\right] \Pi_{j'}^{\bar A} \nonumber \\
    =&\frac{1}{2^{n}}\sum_i \sqrt{\binom{t}{j-i}\binom{t}{j'-i}} |\phi_{j-i}^{(t)}\>\<\phi_{j'-i}^{(t)}|^{\bar A_1 E} (\Pi^{(n-t)}_i)^{\bar A_2}, \label{eq:proj}
\end{align}
and therefore
\begin{align}
    &\Lambda(\Psi, \mathcal{T})\nonumber \\
    =& \frac{1}{2^{k}}\sum_{i,j,j'} \sqrt{\frac{\binom{t}{j-i}\binom{t}{j'-i}\binom{k}{j-\alpha}\binom{k}{j'-\alpha}}{\binom{n}{j}\binom{n}{j'}}} |\phi_{j-i}^{(t)}\>\<\phi_{j'-i}^{(t)}|^{\bar A_1 E} \nonumber\\
    &\quad \otimes (\Pi^{(n-t)}_i)^{\bar A_2}|\phi_{j-\alpha}^{(k)}\>\<\phi_{j'-\alpha}^{(k)}|^{A_1R}|\psi_\alpha\>\<\psi_\alpha|^{A_2} \nonumber \\
    =&  \sum_{i,j,j'} (\Pi^{(n-t)}_i)^{\bar A_2} \otimes |\gamma_{j,i}\>\<\gamma_{j',i}|^{\bar A_1 A_1A_2ER} \nonumber \\
    =& \sum_i (\Pi^{(n-t)}_i)^{\bar A_2} \otimes |\Gamma_i\>\<\Gamma_i|^{\bar A_1 A_1A_2ER}
\end{align}
where
\begin{align}
    |\Gamma_i\> =& \sum_j |\gamma_{j,i}\>, \\
    |\gamma_{j,i}\>^{\bar A_1 A_1A_2ER} =& \sqrt{\frac{\binom{t}{j-i}\binom{k}{j-\alpha}}{2^k\binom{n}{j}}}|\phi_{j-i}^{(t)}\>^{\bar A_1 E}|\phi_{j-\alpha}^{(k)}\>^{A_1R}|\psi_\alpha\>^{A_2}.
\end{align}

As mentioned in Sec.~\ref{subsec:min-entropy}, the min conditional entropy is defined as
\begin{equation}
    H_{\min}(P|Q)_\rho = \sup_{\sigma\ge 0, \Tr\sigma=1}\sup\{\lambda\in \R| 2^{-\lambda}I^P\otimes \sigma^Q \ge \rho^{PQ}\},
\end{equation}
or equivalently, $H_{\min}(P|Q)_\rho=-\log s$ where $s$ is the optimum value of the following semidefinite program (SDP)
\begin{equation}
    s=\inf \Tr\sigma, \quad \text{s.t. }I^P\otimes \sigma^Q \ge \rho^{PQ},\,\sigma\ge0. \label{eq:primal}
\end{equation}
The corresponding dual SDP is
\begin{equation}
    t=\sup \<\rho^{PQ},y^{PQ}\>,\quad \text{s.t. }\Tr_P[y^{PQ}]\le I_Q,\,y\ge 0. \label{eq:dual}
\end{equation}
It is obvious that both primal and dual are strongly feasible, so $s=t$. We can use the following lemma to relate the min-entropy of $\Lambda(\Psi, \mathcal{T})$ to the min-entropy of each $|\Gamma_i\>$.
\begin{lemma}
\label{lem:sum}
Suppose the register $P$ in \zwnn{Eqs.}~\eqref{eq:primal} and \eqref{eq:dual} can be \lhn{decomposed} into \lhn{registers} $P_1$ and $P_2$, and the state $\rho^{PQ}$ has the structure
\begin{equation}
    \rho^{PQ} = \sum_{i=1}^m \Pi_i^{P_1} \otimes \rho_i^{P_2Q}
\end{equation}
where $\Pi_i$ are projectors into disjoint subspaces. Then $s$, the result of the SDP, \zwnn{satisfies} 
\begin{equation}
\frac{1}{m}\sum_i s_i \le  s \le \sum_i s_i
\end{equation}
where $s_i$ is the result \zwnn{of} the SDP \zwnn{for} $\rho_i$.
\end{lemma}

\begin{proof}
We prove the lemma by constructing feasible solutions of the primal and the dual. Let
\[
    \sigma=\sum_i \sigma_i
\]
where $\sigma_i$ is the optimal solution for the primal SDP of $\rho_i$. Then it is natural that $\Tr[\sigma]=\sum_i s_i$, and the condition \lhn{for the SDP} holds because
\begin{equation}
    I^{P_1P_2} \otimes \sigma \ge \sum_i \Pi_i^{P_1} \otimes I^{P_2} \otimes \sigma_i^Q \ge \sum_i \Pi_i^{P_1} \otimes \rho_i^{P_2Q} = \rho.
\end{equation}

For the dual SDP, let
\begin{equation}
    y = \frac{1}{m}\sum_i \frac{\Pi_i^{P_1}}{\Tr[\Pi_i]} \otimes y_i^{P_2Q}
\end{equation}
where $y_i$ is the optimal solution for the dual SDP of $\rho_i$. Then
\begin{equation}
    \<\rho,y\> = \frac{1}{m}\sum_i \Tr\left[\Pi_i\frac{\Pi_i}{\Tr[\Pi_i]}\right]\Tr[y_i\rho_i] = \frac{1}{m}\sum_i s_i,
\end{equation}
and
\begin{equation}
    \Tr_P y = \frac{1}{m}\sum_i \Tr \frac{\Pi_i}{\Tr[\Pi_i]} \Tr_{P_2}[y_i] \le \frac{1}{m}\sum_i I^Q = I^Q.
\end{equation}
\end{proof}
Note that the state $|\Gamma_i\>$ is a pure state, so its min entropy can be calculated using Eq.~\eqref{eq:pure-min}. Therefore we have
\begin{equation}
    H_{\min}(A^*|RE)_{\Gamma_i} = -2\log\left(\sum_j \frac{1}{\sqrt{2^k\binom{n}{j}}}\binom{t}{j-i}\binom{k}{j-\alpha}\right),
\end{equation}
and the value for the corresponding SDP is
\begin{equation}
  \left(\sum_j \frac{1}{\sqrt{2^k\binom{n}{j}}}\binom{t}{j-i}\binom{k}{j-\alpha}\right)^2. 
\end{equation}

Note that $j$ and $i$ should satisfy
\[
    0 \le j-i \le t, \quad 0 \le j-\alpha \le k,
\]
so $\alpha-t \le i \le \alpha + k$, and there are at most $k+t+1$ possible values for $i$. By Lemma~\ref{lem:sum} we have
\begin{equation}
     -\log\kappa \le H_{\min}(A^*|RE)_{\Lambda} \le  -\log\frac{\kappa}{k+t+1}. \label{eq:entropy-bound}
\end{equation}
where
\begin{equation}
    \kappa = \sum_i\left(\sum_j \frac{1}{\sqrt{2^k\binom{n}{j}}}\binom{t}{j-i}\binom{k}{j-\alpha}\right)^2.
\end{equation}
Note that
\begin{align}
    &\frac{1}{\sqrt{2^k\binom{n}{j}}}\binom{t}{j-i}\binom{k}{j-\alpha}\nonumber\\
    \le& 2^{-k/2}\binom{t}{t/2}\binom{k}{k/2}\frac{1}{\sqrt{\min\{\binom{n}{\alpha},\binom{n}{\alpha+k}\}}},
\end{align}
so from Eq.~\eqref{eq:entropy-bound} we have
\begin{align}
    H_{\min}(A^*|RE)_{\Lambda} \ge& n \min\left\{H_b\left(\frac{\alpha}{n}\right),H_b\left(\frac{\alpha+k}{n}\right)\right\}\nonumber\\
    &-2t-k+O(\log n)
\end{align}
for general values of $t$ and $k$ as long as $\alpha$ is linear in $n$. Here $H_b(\cdot)$ is the binary Shannon entropy
\begin{equation}
    H_b(x)\zwnnn{:=} -x\log x -(1-x)\log(1-x), \quad 0 \le x \le 1.
\end{equation}
If  \zwnnn{$t,k = o(n)$ (which is required by the $t^2k^2=o(n)$ condition in Thm.~\Ref{thm:choi}}), this implies $H_{\min}(A^*|RE)_{\Lambda} = \Omega(n)$.

\linghang{
When $\alpha$, $k$ and $t$ are all $O(1)$ and does not depend on $n$, the bound in Eq.~\eqref{eq:entropy-bound} implies that
\begin{equation}
     H_{\min}(A^*|RE)_{\Lambda} \ge \alpha\log n + O(1).
\end{equation}
}

\subsection{Worst-case}\label{app:alternative_worst}
\zwnn{Here} we are interested in the conditional entropy $H_{\min }(A^*|RE)_{\Lambda(\Psi, \mathcal{T})}$ with the initial states $\Psi$ being $|x\>$, $|\mu^\pm_{x,x'}\>$ and $|\nu^\pm_{x,x'}\>$, where
\begin{equation}
    |\mu^\pm_{x,x'}\> = \frac{1}{\sqrt 2}(|x\> \pm |x'\>), \quad|\nu^\pm_{x,x'}\> = \frac{1}{\sqrt 2}(|x\> \pm i|x'\>).
\end{equation}
In other words, the state $|\Psi^{AR}\>$ is one of the above states appended by $|\psi_\alpha\>$, a state with Hamming weight $\alpha$. The reference system $R$ is now trivial, in contrast to the $k$ qubits in Appendix~\ref{app:entropy}. Here the channel $\mathcal{T}$ is the erasure channel over $n-t$ qubits.  Using Eq.~\eqref{eq:proj}, we have
\begin{align}
    &\Lambda(\Psi, \mathcal{T}) \nonumber\\
    =& \frac{1}{\binom{n}{|x|+\alpha}} |x\>\<x|^{A_1} \otimes |\psi_\alpha\>\<\psi_\alpha|^{A_2} \nonumber\\
    & \otimes \sum_i \binom{t}{|x|+\alpha-i}|\phi_{|x|+\alpha-i}^{(t)}\>\<\phi_{|x|+\alpha-i}^{(t)}|^{\bar A_1 E} (\Pi_i^{(n-t)})^{\bar A_2} \nonumber\\
    =& \sum_i (\Pi_i^{(n-t)})^{\bar A_2} \otimes |\gamma_{x,i}\>\<\gamma_{x,i}|, \label{eq:lambda}
\end{align}
where
\begin{equation}
    |\gamma_{x,i}\> = \sqrt{\frac{\binom{t}{|x|+\alpha-i}}{\binom{n}{|x|+\alpha}}}|x\>^{A_1} \otimes |\psi_\alpha\>^{A_2} \otimes |\phi_{|x|+\alpha-i}^{(t)}\>^{\bar A_1 E}.
\end{equation}
Now using Lemma~\ref{lem:sum} and Eq.~\eqref{eq:pure-min}, we have
\begin{equation}
   -\log\left[ \frac{\sum_i\binom{t}{|x|+\alpha-i}^2}{\binom{n}{|x|+\alpha}}\right]  \le  H_{\min }^x \le -\log\left[\frac{\sum_i\binom{t}{|x|+\alpha-i}^2}{\binom{n}{|x|+\alpha}(t+1)}\right], 
\end{equation}
where $H_{\min }^x$ stands for $H_{\min }(A^*|RE)_{\Lambda(\Psi, \mathcal{T})}$ when the initial state is $|x\>\<x|$. This could be further simplified to
\begin{equation}
   -\log\left[ \frac{\binom{2t}{t}}{\binom{n}{|x|+\alpha}}\right]  \le  H_{\min }^x \le -\log\left[\frac{\binom{2t}{t}}{\binom{n}{|x|+\alpha}(t+1)}\right]. 
\end{equation}
When the initial state is $|\mu^\pm_{x,x'}\>$ \lhn{or} $|\nu^\pm_{x,x'}\>$, the state in Eq.~\eqref{eq:lambda} will have the same form, with $|\gamma_{x,i}\>$ replaced by $\frac{1}{\sqrt{2}}(|\gamma_{x,i}\>\pm |\gamma_{x',i}\>)$ and $\frac{1}{\sqrt{2}}(|\gamma_{x,i}\>\pm i|\gamma_{x',i}\>)$ correspondingly. Then
\begin{equation}
   -\log\left[\frac{\chi}{2} \right]  \le  H_{\min }^{x,x'} \le -\log\left[\frac{\chi}{2(t+1)} \right],
\end{equation}
where $H_{\min }^{x,x'}$ stands for $H_{\min }(A^*|RE)_{\Lambda(\Psi, \mathcal{T})}$  when the initial state is one of $|\mu^\pm_{x,x'}\>$ or $|\nu^\pm_{x,x'}\>$, and
\begin{align}
    \chi =& \sum_i\left(\frac{\binom{t}{|x|+\alpha-i}}{\sqrt{\binom{n}{|x|+\alpha}}}+\frac{\binom{t}{|x'|+\alpha-i}}{\sqrt{\binom{n}{|x'|+\alpha}}}\right)^2 \nonumber\\ 
    =& \frac{\binom{2t}{t}}{\binom{n}{|x|-\alpha}}+\frac{\binom{2t}{t}}{\binom{n}{|x'|-\alpha}}+\frac{2\binom{2t}{t+|x|-|x'|}}{\sqrt{\binom{n}{|x|-\alpha}\binom{n}{|x'|-\alpha}}}.
\end{align}

Suppose that in the large $n$ limit $\frac{\alpha}{n}$ and $\frac{\alpha+k}{n}$ are both constants between 0 and 1, we have
\begin{align}
    H_{\min }^x =& n H_b\left(\frac{|x|+\alpha}{n}\right) -2t +O(\log n) \nonumber\\
    \ge& n \min\left\{H_b\left(\frac{\alpha}{n}\right),H_b\left(\frac{\alpha+k}{n}\right)\right\} - 2t + O(\log n),
\end{align}
and
\begin{align}
    &H_{\min }^{x,x'} \nonumber\\
    \ge & n \min\left\{H_b\left(\frac{|x|+\alpha}{n}\right),H_b\left(\frac{|x'|+\alpha}{n}\right)\right\} - 2t + O(\log n) \nonumber\\ 
    \ge & n \min\left\{H_b\left(\frac{\alpha}{n}\right),H_b\left(\frac{\alpha+k}{n}\right)\right\} - 2t + O(\log n).
\end{align}

\section{\zwnnn{Average physical states} in the $U(1)$ case}
\subsection{Choi}
\label{app:average-state}
Following the definitions \zwnn{above}, $\Phi_\text{avg}$ is a joint state on $n$-qubit register $A$ and $k$-qubit register $R$. From Eq.~\eqref{eq:avg-state},
\begin{align}
    &\Phi_\text{avg}^{RA} \nonumber\\
    =& \E_{U\sim H_{\times,U(1)}} U(|\hat\phi^{(k)}\>\<\hat\phi^{(k)}|^{A_1R}\otimes |\psi_\alpha\>\<\psi_\alpha|^{A_2})U^\dagger \nonumber \\
    =& \sum_{j=0}^k \Tr_A[\Pi_j^A(|\hat\phi^{(k)}\>\<\hat\phi^{(k)}|^{A_1R}\otimes |\psi_\alpha\>\<\psi_\alpha|^{A_2})\Pi_j^A] \otimes \frac{\Pi_j^A}{\binom{n}{j}} \nonumber \\
    =& 2^{-k}\sum_{j=\alpha}^{k+\alpha} \Pi_{j-\alpha}^R \otimes \frac{\Pi_j^A}{\binom{n}{j}},
\end{align}
and
\begin{align}
    \Tr_{n-t}\Phi_\text{avg}^{RA}=&2^{-k}\sum_{j=\alpha}^{k+\alpha} \sum_{i=0}^{t} \Pi_{j-\alpha}^R \otimes \Pi_i^E\frac{\binom{n-t}{j-i}}{\binom{n}{j}} \nonumber\\
    =& 2^{-k}\sum_{j=0}^{k} \sum_{i=0}^{t} \Pi_{j}^R \otimes \Pi_i^E\frac{\binom{n-t}{j+\alpha-i}}{\binom{n}{j+\alpha}},
\end{align}
where $E$ refers to the $t$-qubit register that the complementary channel maps to. To get an upper bound for the second term of Eq.~\eqref{eq:bound}, we can replace the minimization over $\zeta$ by an arbitrary fixed $\zeta_0$, which we choose to be the marginal state
\begin{equation}
    \zeta_0^E=\Tr_R \Tr_{n-t}\Phi_\text{avg}^{RA} = 2^{-k}\sum_{i=0}^{t}\Pi_i^E \sum_{j=\alpha}^{k+\alpha}  \frac{\binom{k}{j-\alpha}\binom{n-t}{j-i}}{\binom{n}{j}}.
\end{equation}
We define
\begin{equation}
    \quad \beta_i = 2^{-k}\sum_{j=\alpha}^{k+\alpha}  \frac{\binom{k}{j-\alpha}\binom{n-t}{j-i}}{\binom{n}{j}} = 2^{-k}\sum_{j=0}^{k}  \frac{\binom{k}{j}\binom{n-t}{j+\alpha-i}}{\binom{n}{j+\alpha}},\label{eq:beta}
\end{equation}
so that
\begin{equation}
    \zeta_0^E = \sum_{i=0}^{t}\beta_i\Pi_i^E
\end{equation}

Note that all the states are diagonal, the fidelity is given by
\begin{align}
    &F\left(\Tr_{n-t}\Phi_\text{avg},\frac{I}{2^k}\otimes \zeta_0\right)\nonumber\\
    =& \sum_{j=0}^{k} \sum_{i=0}^{t} \Tr[\Pi_j^R \otimes \Pi_i^A]\sqrt{2^{-k}\beta_i \times 2^{-k}\frac{\binom{n-t}{j+\alpha-i}}{\binom{n}{j+\alpha}}} \nonumber \\ =&
    2^{-k}\sum_{j=0}^{k} \sum_{i=0}^{t} \binom{k}{j}\binom{t}{i}\sqrt{\beta_i \frac{\binom{n-t}{j+\alpha-i}}{\binom{n}{j+\alpha}}} \label{eq:choi-fidelity}
\end{align}

For any nonnegative number $n$ and real number $x$, we define
\[
    x^{\underline n} = x(x-1)\cdots (x-n+1),\quad x^{\overline n} = x(x+1)\cdots (x+n-1).
\]
They are related by
\[
    x^{\ovl n} = (x+n-1)^{\ul n}, \quad x^{\ul n} = (x-n+1)^{\ovl n}
\]
and
\[
    x^{\ul n} = (-x)^{\ovl n} (-1)^n, \quad x^{\ovl n} = (-x)^{\ul n} (-1)^n.
\]

It could be verified that the following binomial theorems hold (by induction on $n$)
\[
    (x+y)^{\underline n} = \sum_{k=0}^n \binom{n}{k} x^{\ul k} y^{\ul{n-k}},\quad (x+y)^{\ovl n} = \sum_{k=0}^n \binom{n}{k} x^{\ovl k} y^{\ovl{n-k}}.
\]
Then from the definition of $\beta_i$ in Eq.~\eqref{eq:beta} we have
\begin{align}
    \beta_i =& 2^{-k}\sum_{j=0}^{k}  \frac{\binom{k}{j}\binom{n-t}{j+\alpha-i}}{\binom{n}{j+\alpha}} \nonumber \\
    =& \frac{1}{n^{\ul t}2^k} \sum_{j=0}^k \binom{k}{j} {(j+\alpha)^{\ul i}(n-\alpha-j)^{\ul {t-i}}} \nonumber \\
    =& \frac{1}{n^{\ul t}2^k} \sum_{j,x} \binom{k}{j} \binom{i}{x}j^{\ul x}\alpha^{\ul {i-x}}(-1)^{t-i}(j-n+\alpha)^{\ovl {t-i}} \nonumber \\
    =& \frac{1}{n^{\ul t}2^k} \sum_{j,x} \binom{k}{j} \binom{i}{x}j^{\ul x}\alpha^{\ul {i-x}}(-1)^{t-i} \nonumber\\
    &\times (j-n+\alpha+t-i-1)^{\ul {t-i}} \nonumber \\
    =& \frac{1}{n^{\ul t}2^k} \sum_{j,x,y} \binom{k}{j} \binom{i}{x}j^{\ul x}\alpha^{\ul {i-x}}(-1)^{t-i}\binom{t-i}{y} \nonumber \\ 
     & \times (j-x)^{\ul y}(x-n+\alpha+t-i-1)^{\ul {t-i-y}} \nonumber \\
    =& \frac{1}{n^{\ul t}2^k} \sum_{j,x,y}\binom{k}{j}j^{\ul{x+y}}  \binom{i}{x}\alpha^{\ul {i-x}}(-1)^{y}\binom{t-i}{y}\nonumber\\
    &\times (-x+n-\alpha-t+i+1)^{\ovl {t-i-y}} \nonumber \\
    =& \frac{1}{n^{\ul t}}  \sum_{x,y} 2^{-(x+y)}(-1)^{y}k^{\ul{x+y}} \binom{i}{x}\binom{t-i}{y}\nonumber\\
    &\times \alpha^{\ul {i-x}}(n-\alpha-x-y)^{\ul {t-i-y}}. 
\end{align}
Suppose that $kt=o( n)$, we can see that the term with $x=x_0$ and $y=y_0$ is of order $(kt/n)^{x_0+y_0}$ times the term with $x=y=0$. Then by keeping the terms $x+y \le 2$ we have the series expansion
\begin{align}
    \beta_i = & \frac{\alpha^{\ul i}(n-\alpha)^{\ul{t-i}}}{n^{\ul t}} \left( 1 +\frac{i k-a k t}{\left(2 a-2 a^2\right) n}\right.\nonumber\\
    & + \left. \frac{\xi_1}{8 (a-1)^2 a^2 n^2}+O\left(\frac{k^3t^3}{n^3}\right)\right),
\end{align}
where $a=\alpha/n$ and
\begin{align}
    \xi_1 =& k [a^2 (k-1) (t-1) t+i^2 (-4 a+k+3)\nonumber\\
    &+i (2 a (2 a (t-1)-(k+1) t+k+3)-k-3)].
\end{align}
Similarly
\begin{align}
    \frac{\binom{n-t}{j+\alpha-i}}{\binom{n}{j+\alpha}} =& \frac{\alpha^{\ul i}(n-\alpha)^{\ul{t-i}}}{n^{\ul t}}\left(1+\frac{j (a t-i)}{(a-1) a n} \right. \nonumber\\
    &\left. +\frac{\xi_2}{2 (a-1)^2 a^2 n^2}+O\left(\frac{k^3t^3}{n^3}\right)\right)
\end{align}
with
\begin{align}
    \xi_2 =& j [a^2 (j-1) (t-1) t+i^2 (-2 a+j+1)\nonumber \\
    &+i (2 a (a (t-1)+j (-t)+j+1)-j-1)].
\end{align}

Therefore,
\begin{align}
    \sqrt{\beta_i \frac{\binom{n-t}{j+\alpha-i}}{\binom{n}{j+\alpha}}} =& \frac{\alpha^{\ul i}(n-\alpha)^{\ul{t-i}}}{n^{\ul t}}\left(1+\frac{(2 j+k) (a t-i)}{4 (a-1) a n} \right. \nonumber\\
    & \left. +\frac{\xi_3}{32 (a-1)^2 a^2 n^2}+O\left(\frac{k^3t^3}{n^3}\right)\right),
\end{align}
where
\begin{align}
    &\xi_3 \nonumber\\
    =& a^2 t \left(4 j^2 (t-2)+4 j ((k-2) t+2)+k (k (t-2)-2 t+2)\right) \nonumber\\
    +&i^2 \left(4 j (-4 a+k+2)+k (-8 a+k+6)+4 j^2\right)\nonumber\\
    +&2 i \left(-4 j^2 (a (t-2)+1)+j (4 a (2 a (t-1)-k t+2)-4)\right) \nonumber\\
    -&2ik\left( (a (-4 a (t-1)+(k+2) t-2 (k+3))+k+3)\right).
\end{align}
Then we can multiply this by $\binom{t}{i}$ and sum over $i$ to have the fidelity
\begin{equation}
    \sum_i \binom{t}{i}\sqrt{\beta_i \frac{\binom{n-t}{j+\alpha-i}}{\binom{n}{j+\alpha}}} = 1-\frac{t(k-2j)^2}{32a(1-a)n^2}+O\left(\frac{k^3t^3}{n^3}\right). \label{eq:fidelity}
\end{equation}
Now we plug this into Eq.~\eqref{eq:choi-fidelity} and obtain
\begin{equation}
    F\left(\Tr_{n-t}\Phi_\text{avg},\frac{I}{2^k}\otimes \zeta_0\right) = 1-\frac{tk}{32a(1-a)n^2}+O\left(\frac{k^3t^3}{n^3}\right),
\end{equation}
and the corresponding purified distance is
\begin{equation}
    P\left(\Tr_{n-t}\Phi_\text{avg},\frac{I}{2^k}\otimes \zeta_0\right) = \frac{\sqrt{tk}}{4n\sqrt{a(1-a)}}\left(1+O\left(\frac{t^2k^2}{n}\right)\right).
\end{equation}
\linghang{
Another interesting case to consider is $\alpha=O(1)$, which could be directly evaluated from Eq.~\eqref{eq:choi-fidelity} for small values of $k$ and $t$. For example, when $k=t=1$, we have
\begin{align}
    &F\left(\Tr_{n-t}\Phi_\text{avg},\frac{I}{2^k}\otimes \zeta_0\right) \nonumber\\
    =& 1+\frac{-\sqrt{2}-2\sqrt{2}\alpha+\sqrt{\alpha(2\alpha+1)}+\sqrt{(\alpha+1)(2\alpha+1)}}{2\sqrt{2}n}\nonumber\\
    &+O(n^{-2}),
\end{align}
and when $k=t=2$,
\begin{align}
    &F\left(\Tr_{n-t}\Phi_\text{avg},\frac{I}{2^k}\otimes \zeta_0\right) \nonumber\\
    =& 1+\frac{-2-2 \alpha +\sqrt{\alpha  (\alpha +1)}+\sqrt{(\alpha +1) (\alpha +2)}}{2 n}+O(n^{-2}).
\end{align}
In both cases the purified distance is $O(n^{-1/2})$.
}
\subsection{Worst-case}
\label{app:worst-avg}
It is easy to see that when $U$ is sampled from $H_{\times,U(1)}$
\begin{equation}
    \E_{U} U(|x\>\<x'| \otimes |\psi\>\<\psi|) U^\dagger =
    \begin{cases}
        0, & x \not= x' \\
        \Pi^{(n)}_{|x|+\alpha}/\binom{n}{|x|+\alpha}, &x = x'.
    \end{cases} \label{eq:avg-rhoxx}
\end{equation}
In the case of $x=x'$, we \zwnnn{trace out $n-t$ qubits of the above state} 
\zwnnn{and obtain} 
\begin{equation}
    \rho^{x,x}_\text{avg} = \sum_{i=0}^t \Pi_i^{(t)}\frac{\binom{n-t}{|x|+\alpha-i}}{\binom{n}{|x|+\alpha}}.
\end{equation}
We wish to show that this is close to some fixed state $\zeta$ independent of $x$. We propose that $\zeta$ is $\rho^{x,x}_\text{avg}$ averaged over $x$,
\begin{equation}
    \zeta = \frac{1}{2^k}\sum_{j=0}^k\Pi_i^{(t)}\frac{\binom{k}{j}\binom{n-t}{j+\alpha-i}}{\binom{n}{j+\alpha}}=\sum_i \beta_i \Pi_i^{(t)},
\end{equation}
where $\beta_i$ is the quantity previously defined in Eq.~\eqref{eq:beta} from Appendix~\ref{app:average-state}. For any $x$, the fidelity is given by
\begin{equation}
    F(\rho^{x,x}_\text{avg},\zeta) = \sum_i \binom{t}{i} \sqrt{\beta_i \frac{\binom{n-t}{|x|+\alpha-i}}{\binom{n}{|x|+\alpha}}}. 
\end{equation}
which is exactly the result in Eq.~\eqref{eq:fidelity} with $j=|x|$. Then we have
\begin{equation}
    \max_x P(\rho^{x,x}_\text{avg},\zeta) = \frac{k\sqrt{t}}{4n\sqrt{a(1-a)}}\left(1+\left(\frac{kt^2}{n}\right)\right).
\end{equation}

\section{\zwnnn{Average physical states} in the $SU(d)$ case}
\label{app:avg}
In order to calculate the \zwnnn{average physical state} over $H_{\times,SU(d)}^n$, we need the following lemma.
\begin{lemma}
For any operator $M$ acting on $(\C^d)^{\otimes n}$,
\begin{equation}
    \E_{V \sim H_{\times,SU(d)}^n} VMV^\dagger = \frac{1}{n!} \sum_{\pi \in S_n} O_\pi M O_\pi^\dagger.
\end{equation}
In other words, the average over $\cU_{\times,SU(d)}^n$ equals to the average over its subgroup $\{O_\pi|\pi \in S_n\}$.
\label{lem:avg}
\end{lemma}
\begin{proof}
We first prove that
\begin{equation}
    \tilde M := \frac{1}{n!} \sum_{\pi \in S_n} O_\pi M O_\pi^\dagger = \bigoplus_{\lambda\vdash (n,d)} M_\lambda \otimes \frac{I_{\cR_\lambda}}{r_\lambda}\label{eq:form}
\end{equation}
for some operators $M_\lambda$. We first divide $\tilde M $ into blocks 
\begin{equation}
    \tilde M = \sum_{\kappa \lambda} M_{\kappa\lambda}
\end{equation}
and then decompose each $M_{\kappa\lambda}$ into tensor product basis of operators
\begin{equation}
    M_{\kappa\lambda} = \sum_{ab} E^{ab}_{\kappa\lambda} \otimes M^{ab}_{\kappa\lambda}
\end{equation}
where $E^{ab}_{\kappa\lambda}$ is a $l_\kappa \times l_\lambda$ dimensional matrix that has a single 1 on position $(a,b)$ and  0 elsewhere. In other words, $E^{ab}_{\kappa\lambda} = |\phi_\kappa^a\>\<\phi_\lambda^b|$ where $\{|\phi_\kappa^a\>\}_a$ and $\{|\phi_\lambda^b\>\}_b$ are some sets of basis on $\cL_\kappa$ and $\cL_\lambda$, respectively. Note that $O_\pi$ could be decomposed as
\begin{equation}
    O_\pi = \bigoplus_{\lambda \vdash (n,d)} I_{\cL_\lambda} \otimes \pi_\lambda
\end{equation}
where $\pi_\lambda$ is an irrep of $S_n$. Since $[O_\pi, \tilde M] = 0$, we have
\begin{equation}
    \pi_\kappa M^{ab}_{\kappa\lambda} = M^{ab}_{\kappa\lambda} \pi_\lambda,
\end{equation}
and by Schur's lemma we know that $M^{ab}_{\kappa\lambda}=0$ for $\kappa\not=\lambda$ and $M^{ab}_{\kappa\lambda}\propto I_{\cR_\lambda}$ for $\kappa=\lambda$. This proves Eq.~\eqref{eq:form}.

Since $O_\pi \in \cU_{\times,SU(d)}^n$ for any $\pi \in S_n$, we can see that
\begin{align}
    \E_{V \sim H_{\times,SU(d)}^n} V M V^\dagger =& \E_{V \sim H_{\times,SU(d)}^n} V O_\pi M O_\pi^\dagger V^\dagger \nonumber \\
    =& \frac{1}{n!}\sum_{\pi\in S_n}\E_{V \sim H_{\times,SU(d)}^n} V O_\pi M O_\pi^\dagger V^\dagger \nonumber\\
    =& \E_{V \sim H_{\times,SU(d)}^n} V \tilde M V^\dagger \nonumber \\
    =& \tilde M,
\end{align}
Note that in the final step we used the fact that $[V, \tilde M]= 0$ implied by Eq.~\eqref{eq:form}. 
\end{proof}

Then we consider the \zwnnn{average physical state} in our encoding defined as
\begin{equation}
    \Psi_{\avg} = \E_{V \sim H_{\times,SU(d)}^n} V(|\psi\>\<\psi| \otimes \rho_\lambda)V^\dagger.
\end{equation}
First note that $\cU_{\times,SU(d)}^{n-1}$ acting on the $n-1$ qudits of $\rho_\lambda$ is a subset of $\cU_{\times,SU(d)}^n$, because for any $V \in \cU_{\times,SU(d)}^{n-1}$, we have $[V,U^{\otimes (n-1)}]=0$, which implies $[I_d \otimes V,U^{\otimes n}]=0$. This indicates that $I_d \otimes V \in \cU_{\times,SU(d)}^n$. We can average over $H_{\times,SU(d)}^{n-1}$ first without changing the result, because it can be absorbed into the Haar measure $H_{\times,SU(d)}^n$. As a result,
\begin{equation}
    \Psi_{\avg} = \E_{V \sim H_{\times,SU(d)}^n} V\left(|\psi\>\<\psi| \otimes \frac{\Pi_\lambda}{l_\lambda r_\lambda}\right)V^\dagger.
\end{equation}

Then by Lemma~\ref{lem:avg} we know that
\begin{align}
    \Tr_{n-t}\Psi_{\avg}   =& \frac{1}{n} \sum_{j=1}^t |\psi\>\<\psi|^{(j)} \Tr_{n-t} \left(\frac{\Pi_\lambda}{l_\lambda r_\lambda}\right)^{(\bar j)} \nonumber\\
    &+ \frac{n-t}{n} \Tr_{n-t-1} \left(\frac{\Pi_\lambda}{l_\lambda r_\lambda}\right)
\end{align}
where the upper index $(j)$ and $(\bar j)$ refers to the $j$-th and the $t-1$ qudits excluding the $j$-th one, respectively. The general formula for calculating partial trace of $\Pi_\lambda$ could be found in Ref.~\cite[Lemma III.4]{christandl2007},
\begin{equation}
    \Tr_{n-t} \left(\frac{\Pi_\lambda}{l_\lambda r_\lambda}\right) = \frac{1}{r_\lambda}\sum_{\mu\nu}c_{\mu\nu}^\lambda \frac{r_\nu}{l_\mu} \Pi_\mu
\end{equation}
where $c_{\mu\nu}^\lambda$ is the Littlewood--Richardson coefficient. Here, $\mu$ enumerates all partitions of $t-1$ and $\nu$ enumerates all partitions of $n-t$.

When $t=1$, we can see that
\begin{equation}
    \Tr_{n-t}\Psi_{\avg} = \frac{1}{n}|\psi\>\<\psi| + \frac{n-1}{n} \frac{I_d}{d}. \label{eq:avg-channel}
\end{equation}
\linghang{$t=2$ may be solvable in a similar way.}

\subsection{Finding the worst-case input}
\label{app:worst-case}
From \lhn{Eq.~\eqref{eq:avg-channel}}, we can define the channel
\begin{equation}
    \cC^{L\to E}(\rho) = \frac{1}{n}\rho + \frac{n-1}{n}\Tr[\rho] \frac{I}{d},
\end{equation}
and we want to find an input state $|\psi\>^{LR}$ such that $\cC^{L\to E}$ has largest distance from a constant channel when acting on $|\psi\>$. To be more precise, we pick the constant channel $\cT_{\zeta_0}$ that corresponds to the maximally mixed state $\zeta_0 = I/d$ (note that this should be independent of $|\psi\>$) and want to study the quantity
\begin{align}
    \epsilon =& P\left((I^R\otimes \cC^{L\to E})(|\psi\>\<\psi|^{LR}), (I^R\otimes \cT_{\zeta_0}^{L\to E})(|\psi\>\<\psi|^{LR})\right) \nonumber \\
    =& P\left(\frac{1}{n}|\psi\>\<\psi|^{LR}+\frac{n-1}{n}\psi^R \otimes \frac{I}{d}, \psi^R \otimes \frac{I}{d}\right)
\end{align}
where
\[
    \psi^R = \Tr_L[|\psi\>\<\psi|].
\]
We can assume in general that
\begin{equation}
    |\psi\> = \sum_i \sqrt{p_i}|i\>|i\>
\end{equation}
(note that the coefficients are made real by redefining the basis of $R$) and the fidelity is given by
\begin{align}
    &f:= F\left(\frac{1}{n}|\psi\>\<\psi|^{LR}+\frac{n-1}{n}\psi^R \otimes \frac{I}{d}, \psi^R \otimes \frac{I}{d}\right) \nonumber \\
    =& \Tr\sqrt{\sqrt{\psi^R \otimes \frac{I}{d}}\left(\frac{1}{n}|\psi\>\<\psi|^{LR}+\frac{n-1}{n}\psi^R \otimes \frac{I}{d}\right)\sqrt{\psi^R \otimes \frac{I}{d}}} \nonumber \\
    =& \Tr\sqrt{A+\delta B}
\end{align}
where
\begin{align}
    A =& \left(\psi^R \otimes \frac{I}{d}\right)^2 = \sum_{i,j} \frac{p_i^2}{d^2} |i\>\<i| \otimes |j\>\<j|,\\
    B =& -A + \sqrt{\psi^R \otimes \frac{I}{d}}|\psi\>\<\psi|^{LR}\sqrt{\psi^R \otimes \frac{I}{d}} \nonumber\\
    =& -A + \frac{p_i p_j}{d}|i\>\<j| \otimes |i\>\<j|
\end{align}
and $\delta = 1/n$. Using matrix calculus \cite[Thm.~3.25]{hiai2014} we can find that
\begin{equation}
    \frac{df}{d\delta} = \frac{1}{2}\Tr[(A+\delta B)^{-1/2} B].
\end{equation}
which equals to 0 at $\delta = 0$. To further calculate the second order derivative, we use the theorem again and obtain
\begin{align}
    \left.\frac{d}{d\delta} (A+\delta B)^{-1/2} \right|_{\delta = 0} =& \frac{1}{2}\sum_{i,j}\frac{d}{p_i}|i\>\<i| \otimes |j\>\<j| \nonumber\\
    &- \frac{d^2}{p_i+p_j}|i\>\<j| \otimes |i\>\<j|
\end{align}
which shows that
\begin{equation}
    \left.\frac{d^2 f}{d\delta^2}\right|_{\delta=0} = \frac{1}{4}\sum_i p_i - \frac{d}{2}\sum_{i,j} \frac{p_i p_j}{p_i+p_j} = \frac{1}{4} - \frac{d}{2}\sum_{i,j} \frac{p_i p_j}{p_i+p_j}.
\end{equation}
Then we can use the Taylor series
\begin{equation}
    f = 1 + \left.\frac{1}{2n^2} \frac{d^2 f}{d\delta^2}\right|_{\delta=0} + O(n^{-3})
\end{equation}
and obtain
\begin{align}
    \epsilon =& \frac{1}{n}\sqrt{\frac{d}{2}\sum_{i,j} \frac{p_i p_j}{p_i+p_j} - \frac{1}{4}} + O(n^{-2}) \nonumber \\
    \le& \frac{1}{n}\sqrt{\frac{d}{4}\sum_{i,j} \left(\frac{p_i }{2} + \frac{p_j}{2}\right) - \frac{1}{4}}+ O(n^{-2}) \nonumber \\
    =& \frac{\sqrt{d^2-1}}{2n}+O(n^{-2}).
\end{align}
Note that we have used the fact that $\frac{2p_ip_j}{p_i+p_j} \le \frac{p_i+p_j}{2}$ and this maximum is achieved only at $p_1=p_2=\cdots =p_d=1/d$.

\end{document}